\documentclass[reqno, 11pt]{amsart}

\date{November 24, 2020}
\usepackage{amsmath,amssymb,amsthm,amsfonts}
\usepackage{color}
\usepackage{hyperref}
\numberwithin{equation}{section}

\newtheorem{theorem}{Theorem}[section]
\newtheorem{lemma}[theorem]{Lemma}
\newtheorem{prop}[theorem] {Proposition}
\newtheorem{cor}[theorem]  {Corollary}

\theoremstyle{definition}
\newtheorem{assumption}{Assumption}

\theoremstyle{remark}
\newtheorem{remark}{Remark}

\newcommand{\N}{\mathbb{N}}
\newcommand{\R}{\mathbb{R}}
\newcommand{\Z}{\mathbb{Z}}
\newcommand{\C}{\mathbb{C}}
\newcommand{\Q}{\mathbb{Q}}
\renewcommand{\P}{\mathbb{P}}
\newcommand{\E}{\mathbb{E}}
\newcommand{\dd}{\,\mathrm{d}} %integration d
\newcommand{\eps}{\varepsilon}

\newcommand{\smfrac}[2]{\textstyle{\frac {#1}{#2}}}

\newcommand{\vect}[1]{\boldsymbol{#1}}

\DeclareMathOperator{\supp}{supp} % support 

\DeclareMathOperator{\var}{var}  %variance
\newcommand{\be}{\begin{equation}}
\newcommand{\ee}{\end{equation}}
\newcommand{\ba}{\begin{equation} \begin{aligned}}
\newcommand{\ea}{\end{aligned}\end{equation}}
\newcommand{\bes}{\begin{equation*}}
\newcommand{\ees}{\end{equation*}}

\newcommand{\cl}{\ssup{R}}

\newcommand{\ssup}[1] {{\scriptscriptstyle{({#1}})}}
\def\1{{\mathchoice {1\mskip-4mu\mathrm l}      % Blackboard bold 1
{1\mskip-4mu\mathrm l}
{1\mskip-4.5mu\mathrm l} {1\mskip-5mu\mathrm l}}}
\newcommand{\e}   {{\operatorname e }}
\newcommand{\heap}[2]{\genfrac{}{}{0pt}{}{#1}{#2}}

\begin{document}

%%%%%%%%%%%%%%
% Front matter %% 
\title[Distribution of cracks in a chain of atoms at low temperature]{\Large Distribution of cracks \\\medskip in a chain of atoms at low temperature}

\author{Sabine Jansen}
\address{Mathematisches Institut, Ludwig-Maximilians-Universit\"at M\"unchen, Theresienstra\ss e 39, 80333 M\"unchen, Germany}
\email{jansen@math.lmu.de}
\author{Wolfgang K{\"o}nig}
\address{Weierstrass Institute Berlin, Mohrenstr. 39, 10117 Berlin and Technische Universit{\"a}t Berlin, Str. des 17. Juni, 10623 Berlin, Germany}
\email{koenig@wias-berlin.de}
\author{Bernd Schmidt}
\address{Institut f{\"u}r Mathematik, Universit{\"a}t Augsburg, Universit{\"a}tsstr. 14, 86159 Augsburg, Germany}
\email{bernd.schmidt@math.uni-augsburg.de}
\author{Florian Theil}
\address{Mathematics Institute, University of Warwick, Coventry, CV4 7AL, UK}
\email{F.Theil@warwick.ac.uk}
%

%%%%%%%%
\begin{abstract}
	We consider a one-dimensional classical many-body system with interaction potential of Lennard-Jones type in the thermodynamic limit at low temperature $1/\beta\in(0,\infty)$.  
	The ground state is a periodic lattice.  We show that when 
	the density is strictly smaller than the density of the ground state lattice, the  system with $N$ particles fills space by alternating approximately crystalline domains (clusters) with
	empty domains (voids) due to cracked bonds.  The number of domains is of the order of $N\exp(- \beta e_\mathrm{surf}/2)$ with $e_\mathrm{surf}>0$ a surface energy. 
	
	For the proof, the system is mapped to an effective model, which is a low-density lattice gas of defects. The results require conditions on the interactions between defects. We succeed in verifying these conditions for next-nearest neighbor interactions, applying recently derived uniform estimates of correlations. \\

	\noindent \emph{Keywords}: Equilibrium statistical mechanics; atomistic models of elasticity; fracture; lattice gas of defects. \\
	
	\noindent \emph{Mathematics Subject Classification (2010)}:
	 82B21, %[Stat. Mech: Continuum models]
	  74B20, %[nonlinear elasticity]
	   74G65, %[Energy minimization],
	  60F10. %large deviations
\end{abstract}

\maketitle

\section{Introduction}\label{sec:intro} 

A fundamental problem in statistical and solid mechanics is to gain insight into the structure of matter and to derive material properties from basic atomistic interaction models. A complete  theoretic understanding of why atoms at low energy arrange in (almost) periodic patterns and how defects form appears to be out of reach in full generality to date. In view of this state of affairs, recent years have witnessed remarkable progress under simplifying assumptions and shed light on a number of important model cases. 

A basic, yet non-trivial model problem is given by one-dimensional chains of atoms. Assuming that particles interact via a classical pair interaction potential such as the Lennard-Jones potential, their crystallization in ground states at zero temperature has been well understood since the pioneering contributions \cite{Ventevogel:78,gardner-radin79,Radin:83,RadinSchulman:83}. Even results in a purely quantum mechanical framework have been obtained more recently in \cite{BlancLeBris:02}. Allowing for configurations whose energy is slightly larger while keeping the temperature to be zero, one is led to considering chains of atoms that may undergo fracture. Also this regime is well examined by now, in particular for systems with nearest neighbor (NN) and next-to-nearest neighbor (NNN) interactions, see \cite{braides-cicalese07,scardia-schloem-zanini11,hudson2013}. In contrast to the case of pure nearest neighbor interactions such as \cite{truskinovsky96,braides-dalmaso-garroni99,BraidesGelli:02}, such models show a non-trivial competition between NN bonds lying in the convex region and NNN bonds occupying the concave region of the interaction potential. This, in particular, leads to nontrivial surface corrections and boundary layers. Extension to more general finite range interactions are considered in \cite{BraidesLewOrtiz:06,schaeffner-schloemerkemper2018}, a detailed analysis of boundary layers and surface energies is provided in \cite{jkst19}. 

We also mention that, by way of contrast, much less is known in higher dimensions, still within the zero temperature regime. We refer to \cite{HeitmannRadin:80,Radin:81,Sueto:05,Theil:06,FlatleyTheil:15} for crystallization results for specific pair potentials and to \cite{AlicandroFocardiGelli:00,FriedrichSchmidt:14,FriedrichSchmidt:15a,FriedrichSchmidt:15b} for partial results relating atomistic models to a corresponding variational continuum Griffith functional. 

Still in a one-dimensional setting, in our recent contribution \cite{jkst19} we proved that at small but non-zero temperature a chain of atoms under positive pressure is well approximated by the corresponding ground state configuration. In the present article we show that the picture is rather different for a chain of atoms in thermal equilibrium at small non-zero temperature and vanishing pressure. Indeed, at positive temperature, no matter how small, there is no global crystallization in one-dimensional systems for typical interaction potentials. Yet we will see that an alternating pattern of crystalline clusters and cracks emerges whose statistics can be precisely described in terms of an associated surface/defect energy. 

We assume that atoms interact via a Lennard-Jones type potential, energy minimizers have their interatomic spacing (in the bulk) equal to the minimizer $a>0$ of a Cauchy-Born energy density. Thermal equilibrium is investigated within the  framework of classical equilibrium statistical mechanics \cite{ruelle-book69, presutti-book}. This means that we study families of probability measures indexed by the number $N$ of atoms, the length $L$ of the chain, and a positive parameter $\beta>0$ called inverse temperature. Each configuration has  probability weight proportional to $\exp( - \beta U)$, with $U$ the energy of the configuration. Our results are formulated as asymptotic estimates for finite $N,L,\beta$ but they should be read with the following limits in mind: thermodynamic limit $L,N\to \infty$ at fixed $\ell = L/N$ first, low-temperature limit $\beta \to \infty$ second. 

Our main result (Theorem~\ref{thm:lowdens}) roughly says that for elongated chains i.e. $L/N=\ell>a$, and large $\beta$, the chain of atoms typically fills space by alternating approximately crystalline domains with empty regions of space. By approximately crystalline domain we mean a cluster of consecutive atoms with spacing approximately equal to the optimum ground state spacing $a$, except at both ends of the cluster where boundary layers may form. Empty regions of space occur when two consecutive atoms $x_j \leq x_{j+1}$ are separated by a large spacing $z_j = x_{j+1} - x_j$, which we call \emph{gaps}, \emph{voids}, \emph{cracks}, or \emph{broken bonds}.   We identify a surface free energy $e_\mathrm{surf}^\ssup{R}(\beta)$ for a threshold parameter $R$ which corresponds to a critical length beyond which atomic bonds are considered broken and whose $R$-dependence vanishes for $\beta \to \infty$ and show that: (i) The number of cracks (hence also the number of crystalline clusters) is of  the order of $N\sqrt{\ell-a}\, \exp( - \beta e_\mathrm{surf}^\ssup{R}(\beta) /2)$, (ii) the typical length of a crack is of the order of $\sqrt{\ell-a}\, \exp( \beta e_\mathrm{surf}^\ssup{R}(\beta)/2)$, and (iii) the typical number of atoms in a crystalline cluster is of the order of $\exp( \beta e_\mathrm{surf}^\ssup{R}(\beta)/2)/\sqrt{\ell-a}$. A posteriori we will see that this asymptotic behavior is in fact essentially independent of the choice of $R$ for large $\beta$. 

In particular, the number of cracks is not bounded but instead proportional to the number of atoms in the chain, moreover each crack is of microscopic length even though the length is exponentially large in $\beta$. This behavior is similar to one-dimensional Ising chains with nearest neighbor interaction at low temperature \cite{schonmann-tanaka1990} or with Kac interactions and small Kac parameter \cite{cassandro-orlandi-presutti1993,cassandro-merola-presutti2018}. For the Ising model intervals of positive magnetization and negative magnetization play a role analogous to our crystalline clusters and empty domains. 

Let us briefly explain what makes our result demanding. One-dimensional Gibbs measures with finite-range interactions are easily treated with transfer operators and infinite-dimensional versions of Perron-Frobenius theory; absence of phase transitions, analyticity of thermodynamic potentials, and decay of correlations follow right away \cite{ruelle-book69}. The challenge taken up here is to characterize how the objects whose existence is trivially guaranteed by Perron-Frobenius theory depend on the inverse temperature $\beta$. In principle it should be possible to do so by studying the $\beta$-dependency of the transfer operator. However, as pointed out by Cassandro, Merola, and Presutti already in the context of the Ising model, ``to carry out the whole program along these lines looks maybe possible but not easy at all''  \cite{cassandro-merola-presutti2018}. 

Therefore we follow a different route and instead map the chain of atoms to an effective model which is a weakly interacting lattice gas of defects, see Sections~\ref{sec:heuristics2} and~\ref{sec:effective-model}.  Sites $j$ of the lattice correspond to labels of nearest-neighbor spacings. A particle or defect is present in the effective model if the original bond is broken, i.e., the gap is large. The presence of a defect comes with a small weight $q_\beta>0$, the {\em effective activity}, which is related to the free energy of formation of the defect (see Remark~\ref{rem:defect-eof}).  Defects separated by $k$ particles have an effective interaction $V_\beta(k)$ which can be written in terms of the energy of the particles enclosed between cracks (Eqs.~\eqref{eq:interact} and~\eqref{eq:veff}). 

A principal difficulty is to show that the effective interaction between defects can be neglected. This is achieved with Theorem~\ref{thm:codec}, which is our main technical result, for chains of atoms with NN and NNN interactions. It crucially relies on bounds on the decay of correlations. Such estimates are highly non-trivial for interactions beyond nearest neighbors since, as alluded to above, boundary layers will form and give rise to non-trivial surface energy contributions. In the case of NN and NNN interactions, a sufficiently strong result has recently been established in \cite[Theorem 2.11]{jkst19}. With a view to the interesting question if our conclusions extend to more general finite range interactions, we formulate our results in a way that allows for such an adaption subject to sufficiently good decay of correlations estimates becoming available. 

The article is organized as follows. Section~\ref{sec:model-results} describes the model and main results. Section~\ref{sec:heuristics} explains key proof ingredients, namely equivalence of ensembles and, on a heuristic level, the effective lattice gas of defects. In addition it proposes an alternative scenario replacing Theorem~\ref{thm:lowdens} when some of our assumptions fail.  Section~\ref{sec:aux} analyzes in detail the corresponding class of effective models. These general results are applied in Section~\ref{sec:pressure} to the constant-pressure ensemble, the results for the canonical ensemble are deduced in Section~\ref{sec:canonical}.

\section{Model and main results}\label{sec:model-results}

\subsection{Model} 
Consider $N$ particles on a line, with positions $0=x_1<x_2<\cdots <x_N=L$ and spacings $z_j=x_{j+1} - x_j$. Neighboring particles and next-to-nearest neighbors interact via a pair potential 
$v\colon[0,\infty) \to \R\cup \{\infty\}$ which is repulsive for short distances (in fact we shall assume the existence of a hard core) and attractive for spacings larger than a unique energy minimizing bond length. The precise assumptions are collected in Section~\ref{sec:ass} below. The total {\em energy of a configuration} $(x_1, \ldots, x_N)$ is 
\[
	\sum_{i=1}^{N-1} v(x_{i+1}-x_i) + \sum_{i=1}^{N-2} v(x_{i+2}-x_i). 
\] 
Since our analysis extends in a straightforward way to more general interactions involving a finite number of $m \in \N$ particles subject to improved estimates on correlations being available, cf.\ the discussion in Section~\ref{sec:intro}, we more generally consider 
\[
	U_N^{\ssup m}(z_1,\ldots,z_{N-1})= \sum_{\heap{1\leq i<j \leq N}{|j-i|\leq m}} v(x_j-x_i)= \sum_{\heap{1\leq i<j \leq N}{|j-i|\leq m}} v(z_i+\cdots + z_{j-1}). 
\]
We allow for $m \in \N \cup\{\infty\}$ and specify explicitly whenever $m < \infty$ or  $m=2$ is exploited, but will sometimes drop the superscript $m$ so as to lighten notation. The canonical partition function at inverse temperature $\beta>0$, with one particle pinned at $x_1=0$ and another at $x_N=L$ is 
\[
	Z_N^{\ssup m}(\beta,L):= \int_{\R_+^{N-1}} \e^{-\beta U_N^{\ssup m}(z_1,\ldots,z_{N-1})}\1_{\{z_1+\cdots + z_{N-2}\leq L\}} \dd z_1\cdots \dd z_{N-2}
\]
where we put $z_{N-1}= L-\sum_{j=1}^{N-2} z_j$. The canonical Gibbs measure is the probability measure $\P_{N,L}^\ssup{m,\beta}$ on $\Omega_{N,L}:= \{\boldsymbol{z}\in \R_+^{N-1}\mid z_1+\cdots + z_{N-1}= L\}$  defined by 
\[
	 \P_{N,L}^\ssup{m,\beta}(A) = 
	\frac{1}{Z_N^{\ssup m}(\beta,L)}\int_{\Delta_{N,L}} \e^{-\beta U_N^{\ssup m}(z_1,\ldots,z_{N-1})}\1_A (z_1,\ldots,z_{N-1}) \dd z_1\cdots \dd z_{N-2}
\]
where $\Delta_{N,L}\subset [0,L]^{N-2}$ is the simplex $z_1+\cdots + z_{N-2}\leq L$ and $z_{N-1}:= L- \sum_{j=1}^{N-2} z_j$. 
Fix an average spacing $\ell>0$. The Helmholtz free energy per particle is
\[ \label{eq:fmlimit}
	f^{\ssup m}(\beta,\ell):= -  \lim_{N\to \infty} \frac{1}{\beta N}\log Z^{\ssup m}_{N}(\beta,\ell N). 
\]
 
The existence of the limit~\eqref{eq:fmlimit} and some basic properties are well-known \cite[Chapter 3]{ruelle-book69}. Moreover $\ell\mapsto f^{\ssup m}(\beta,\ell)$ is convex and continuously differentiable~\cite{dobrushin-minlos67}, see also~\cite[Chapter 3.4.8]{ruelle-book69}. For one-dimensional systems and the pair potentials under consideration, there is no phase transition, and $\ell\mapsto f^{\ssup m}(\beta,\ell)$ is strictly convex and analytic~\cite{gallavotti-miraclesole70,dobrushin73,dobrushin74, cassandro-olivieri81}. The above references refer to the model with $m=\infty$; for finite $m$, such results are easily proven using transfer operators~\cite[Chapter 5.6]{ruelle-book69}.
The {\em pressure} is defined as 
\be \label{eq:pfrel}
	p^{\ssup m}(\beta,\ell) = - \frac{\partial f^{\ssup m}}{\partial \ell}(\beta,\ell).
\ee
In addition to the free energy and pressure, we provide results on the distribution of interparticle spacings. We investigate the following objects. 
Let $R>0$ be some large truncation parameter. We refer to spacings $z_j\geq R$ as \emph{cracks} and to groups of particles enclosed between consecutive cracks as \emph{clusters}. Let 
\begin{equation}\label{clusternumber}
M_N(z_1,\ldots,z_{N-1}) := \#\{i \in \{1,\ldots,N-1\}\colon z_i \geq R\}+1
\end{equation} 
be the number of clusters. For $M_N = n+1$ let $1\leq i_1<\cdots <i_n\leq N-1$ be the indices $i$ for which $z_i\geq R$. We also set $i_0=0$ and $i_{n+1} = N$. Let 
\be\label{empdist-card-length}
	\nu_N = \frac{1}{M_N}\sum_{k=1}^{M_N}\delta_{i_k - i_{k-1}},\quad   
	\widehat \nu_N = \frac{1}{M_N-1} \sum_{j=1}^{N-1} \1_{[R,\infty)}(z_j) \delta_{z_j- R} 
\ee
be the empirical distributions of the cluster cardinalities and of the crack lengths (minus $R$); note that they are probability measures.

\subsection{Assumptions} \label{sec:ass}

In this section we introduce and discuss the four assumptions on the pair potential under which we will be working throughout this article. Their role is threefold. First of all, they ensure that standard results from statistical mechanics concerning the existence of the thermodynamic limit, continuity of the pressure and absence of phase transitions in dimension one apply. Second, they ensure the periodicity of ground states and allow for a transfer of the low-temperature asymptotics from our previous article~\cite{jkst19}. Third, for average spacings $\ell$ larger than the ground state periodicity $a$, they allow us to estimate interactions across cracks and to show that cracks do not aggregate.

\begin{assumption}[on the interaction potential] \label{ass:v1}
The pair potential $v : (0, \infty) \to \R \cup \{+ \infty\}$ with hard core radius $r_{\rm hc} > 0$ is equal to $+ \infty$ on $(0, r_{\rm hc}]$ and a $C^2$ function on $(r_{\rm hc}, \infty)$. There exist $r_{\rm hc} < z_{\min} < z_{\max} < 2 z_{\min}$ and $\alpha_1, \alpha_2 > 0$, $s > 2$ such that the following holds. 
\begin{itemize}
\item[(i)] \emph{Shape of $v$}: 

$z_{\max}$ is the unique minimizer of $v$ and satisfies $v(z_{\max})<0$. Furthermore, $v$ is decreasing on $(0,z_{\max})$ and  increasing and non-positive on $(z_{\max},\infty)$. 

\item[(ii)] \emph{Growth of $v$}: 

$v(z) \ge - \alpha_1 z^{-s}$ for all $z > 0$ and $v(z) + v(z_{\max}) - 2 \alpha_1 \sum_{n=2}^\infty (n z)^{-s} > 0$ for all $z < z_{\min}$.  

\item[(iii)] \emph{Shape of $v''$}:

 $v''$ is decreasing on $[z_{\min}, z_{\max}]$ and increasing and non-positive on $[2 z_{\min}, \infty)$.

\item[(iv)] \emph{Growth of $v''$}:

 $v''(z) \ge -\alpha_2 z^{-s-2}$ for all $z > r_{\rm hc}$ and $v''(z_{\max}) + \sum_{n=2}^\infty n^2 v''(n z_{\min}) > 0$. 
 
\item[(v)] \emph{Behavior near $r_\mathrm{hc}$}: 

 $\lim_{r\searrow r_\mathrm{hc}} v(r) = \infty$. 
 
\item[(vi)] \emph{Size of $r_\mathrm{hc}$}: 
 
 $v(r)\leq 0$ for all $r\geq 2r_\mathrm{hc}$. 
\end{itemize}
\end{assumption} 

\noindent Assumptions~\ref{ass:v1}(i)--(v) are rather generic conditions on a pair potential with hard core. They are imported from~\cite{jkst19} and we refer to~\cite{jkst19} for a thorough discussion of these assumptions. They also allow us to estimate interactions across cracks: Indeed, under Assumption~\ref{ass:v1}(i)--(v), there exists a constant $C\geq 0$ such that, for  any $N\in\N$, $R\geq z_{\max}$, $m\in\N\cup\{\infty\}$, and $z_1,\ldots,z_{N-1}$ satisfying $z_k \geq R$ for some $k\in \{1,\ldots,N-1\}$,
\be \label{eq:across}
	- \frac{C}{z_k^{s-2}}\leq  \sum_{\heap{i,j\colon 1\leq i \leq k< j \leq N-1}{|j-i|\leq m}}v(z_i+\cdots + z_{j-1}) \leq 0.
\ee
We leave the elementary proof to the reader. Assumption~\ref{ass:v1}(vi), which relates the hard core radius to the full repulsive zone of $v$, is a mild technical assumption which enters in the proof of Theorems~\ref{thm:codec} and \ref{thm:ge} below. In view of the fact that typical next-to-nearest neighbor bonds are attractive it might be achieved upon enlarging $r_\mathrm{hc}$ while keeping the essential properties of the model. For compactly supported potentials, interactions across cracks vanish if $R$ is longer than the interaction range. Our next assumption ensures this, more generally, it ensures that the entropic push for large crack lengths wins over the attractive part of the interaction; it enters in Lemma~\ref{lem:gapintegral} below.
\medskip

\begin{assumption}[on the truncation parameter] \label{ass:R} 
	The truncation parameter $R\geq z_{\max}$ is so large that it satisfies 
	\begin{itemize} 
		\item $R\geq \sup\supp(v)$ if $v$ is compactly supported,
		\item $C/R^{s-2} < e_\mathrm{surf}/2$ otherwise, with $C>0$ as in Eq.~\eqref{eq:across} and $e_\mathrm{surf}>0$ as in Eq.~\eqref{eq:surface} below.
	\end{itemize} 
\end{assumption} 
\noindent In fact, as our analysis in the second case does not make use of an infinite interaction range, the second condition is sufficient for compactly supported potentials as well. We distinguish the case $R\geq \sup\supp(v)$ as it allows for an approximation with an \emph{ideal} lattice gas. 

	We close this subsection commenting on the role of the interaction parameter $m$. A restriction to finite range $m<\infty$ is natural for compactly supported potentials with a hard core. Indeed, if $v(r)=\infty$ for $r\leq r_\mathrm{hc}$ for some $r_\mathrm{hc}>0$ and $v(r) = 0$ for $r\geq R^*$, then any configuration $z=(z_1,\dots,z_{N-1})$ with finite energy satisfies $v(z_i+\cdots + z_{j-i})=0$ whenever $|j-i|r_\mathrm{hc}> R^*$. Hence, $U_N^{\ssup{\lceil R^*/r_\mathrm{hc}\rceil}}(z)=U_N^{\ssup{\infty}}(z)$. For $v$ with unbounded range a restriction to finite (and in fact small $m$) is quite common in atomistic models of solid state physics and indeed less restrictive than a truncation of the potential $v$ itself as leading order contributions to crack energies are still kept. As our main theorems are proven for NN and NNN interactions, we introduce the following assumption. 
	
\begin{assumption}[on the interaction parameter] \label{ass:m} 
	Suppose that $m=2$.  
\end{assumption} 

	While most auxiliary results apply to $m \ge 3$ as well, Assumption~\ref{ass:m} enters in the proof of Theorem~\ref{thm:codec} below, where we need a good control on the $\beta$-dependence of correlations~\cite[Theorem 2.11]{jkst19}. Improved correlation bounds for general interaction range may help get rid of the restriction, but this is beyond this article's scope.

\subsection{Results}\label{sec-results}

In this section, we formulate our two main results on the large-$N$ behavior of the $N$-particle system  at low temperature: Theorem~\ref{thm:fe} on the free energy and the pressure, and  Theorem~\ref{thm:lowdens} on the statistics of cluster sizes, crack lengths, and number of clusters. 

First, we need to recall results from \cite{jkst19}. 
Let 
\[
	E_N = \inf_{\R_+^{N-1}} U_N
\]
be the $N$-particle ground state energy. The following limits exist:
\be \label{eq:surface}
	e_0 = \lim_{N\to \infty}\frac{E_N}{N}\in(-\infty,0)\qquad\mbox{and}\qquad e_\mathrm{surf} =\lim_{N\to \infty} (E_N - N e_0)\in(0,\infty).
\ee
The {\em ground state energy per particle} $e_0$ is characterized with the help of the {\em Cauchy-Born density} 
\[
	W(r):= \sum_{k=1}^m v(kr).
\]
as $e_0 = \inf_{r>0} W(r) = W(a)$ with $a\in (z_\mathrm{min},z_\mathrm{max})$ the unique global minimizer of $W$ \cite[Section 2.1]{jkst19}. The {\em surface energy} $e_{\rm surf}$ accounts for boundary layers at the end of long chains. The reader may also think of $e_\mathrm{surf}$ as the energy of a defect consisting of a large spacing $z_j$, see Remark~\ref{rem:defect-eof} below. 

For positive temperature, analogous quantities and assertions are collected in the following proposition. The truncated partition function appearing on the left-hand side of \eqref{e0Rdef} will play an important role in the present article.

\begin{prop}[\cite{jkst19}]\label{prop-oldresults}
Under Assumptions~\ref{ass:v1}(i)--(v), for every $\beta>0$ and $0 \le p < |v(z_{\max})|/z_{\max}$, there are uniquely defined quantities $g_\mathrm{surf}^\ssup{R}(\beta,p)$, $g^\ssup{R}(\beta,p)$ such that, as $N\to \infty$, 
\begin{equation}\label{e0Rdef}
	- \frac{1}{\beta} \log \Bigl( \int_{[0,R]^{N-1} } \e^{-\beta [U_N(z_1,\ldots,z_{N-1})+p\sum_{j=1}^{N-1}z_j]} \dd z_1\cdots \dd z_{N-1}\Bigr) 
		= N g^\ssup{R}(\beta,p) + g_\mathrm{surf}^\ssup{R}(\beta,p) + o(1). 
\end{equation}
Moreover, writing $e_\mathrm{surf}^\ssup{R}(\beta) = g_\mathrm{surf}^\ssup{R}(\beta,0) $, $e_0^\ssup{R}(\beta) = g^\ssup{R}(\beta,0)$, 
\begin{equation}\label{e0betalimits}
\lim_{\beta\to \infty} e_0^\ssup{R}(\beta) = e_0\qquad\mbox{and}\qquad\lim_{\beta\to \infty} e_\mathrm{surf}^\ssup{R} (\beta) = e_\mathrm{surf}. 
\end{equation}
\end{prop}

In particular, the $R$-dependence vanishes in the zero-temperature limit. Some technical remarks are in order. Indeed, Sections~2.2 and~2.3 in~\cite{jkst19} assume a fixed positive pressure constant $p>0$. Those results extend to $p=0$ or temperature-dependent pressures $p=p_\beta\to 0$ if the integration is restricted to compact intervals $z_j\in [0,R]$. Indeed the positivity of $p$ is only needed to ensure exponential tightness, see~\cite[Lemma~5.1 and~5.3]{jkst19}. But exponential tightness comes for free in compact spaces, the condition $p>0$ is no longer needed. 

Our first result concerns the asymptotics of the free energy and the pressure as $\beta\to \infty$ at fixed elongation $\ell$. 

\begin{theorem}[Free energy density and pressure for $\beta\to \infty$ at fixed $\ell>z_{\min}$] \label{thm:fe}
Under Assumptions~\ref{ass:v1}--\ref{ass:m}:
\begin{enumerate}
		\item [(a)] There exists $\ell^*<a$ such that for all $\ell\in(\ell^*,a)$, 
		\[
		\lim_{\beta\to \infty} f(\beta,\ell) = W(\ell)>e_0\qquad\mbox{and}\qquad\lim_{\beta\to\infty}p(\beta,\ell)= -W'(\ell)>0.
		\]
		\item [(b)] If $\ell >a$, then, as $\beta\to\infty$, 
			\begin{eqnarray*} 
				 f(\beta,\ell) &= &e_0^\ssup{R}(\beta) - \frac 2 \beta \sqrt{\ell-a}\,\e^{-\beta e_\mathrm{surf}^\ssup{R}(\beta)/2} \, (1+o(1)), \\
				 p(\beta,\ell)& =& \frac1{\beta \sqrt{\ell - a} }\, \e^{- \beta e_\mathrm{surf}^\ssup{R}(\beta)/2}(1+o(1)). %\label{pressasy}
			\end{eqnarray*}  
In particular, $\lim_{\beta\to \infty} f(\beta,\ell) = W(a) = e_0$.
	\end{enumerate} 
\end{theorem} 

The theorem is proven in Section~\ref{sec:canoproof1}. The leading-order asymptotic behavior of the free energy density in both (a) and (b) is written more succinctly with the convex hull $W^{**}$ of $W$ as 
$$
	\lim_{\beta\to \infty} f(\beta,\ell) = W^{**}(\ell) = 
	\begin{cases} 
		W(\ell), &\quad \ell \in (\ell^*,a),\\
		W(a), &\quad \ell \geq a. 
	\end{cases} 
$$
We remark that in view of our general assumptions on $v$ we cannot expect the Cauchy-Born rule to hold near $r_\mathrm{hc}$, so that $\ell^* > r_\mathrm{hc}$ in general.

Our second result describes in detail the distribution of cracks for elongated chains $\ell>a$. The case $\ell<a$ corresponds to positive pressure and was already analysed in detail in \cite{jkst19}. Define
\be \label{eq:qldef}
	q_{\beta,\ell} = \frac{\exp(-\beta e_\mathrm{surf}^\ssup{R}(\beta))}{\beta p(\beta,\ell)} = \sqrt{\ell - a}\,\e^{-\beta e_\mathrm{surf}^\cl(\beta)/2}\bigl(1+o(1)\bigr) .
\ee
For simplicity we suppress the $R$-dependence from the notation for $q_{\beta,\ell}$. We let $\mathrm{Geom}(p)$ denote the probability measure on $\N$ with probability weights $p (1-p)^{k-1}$, and $||\cdot||_\mathrm{TV}$ the total variation distance between two probability measures, i.e., $||\mu - \nu||_\mathrm{TV} = \sup_{A}|\mu(A) - \nu(A)|$. 
Note that both $p(\beta,\ell)$ and $q_{\beta,\ell}$ behave as $\exp(-\beta e_\mathrm{surf}/2 +o(\beta))$. 

\begin{theorem}[Distribution for $\beta\to\infty$ at fixed $\ell>a$] \label{thm:lowdens}
Suppose that Assumptions~\ref{ass:v1}--\ref{ass:m} hold true. 
	Fix $\ell>a$. Then there exist  $\delta^\ssup{i}_\beta>0$ with $\lim_{\beta\to \infty} \delta^\ssup{i}_\beta =0$, $i=1,2$, and  $\beta_0\geq 0$  such that for all $\beta\geq\beta_0$, 
	\begin{align*}
		\limsup_{N\to \infty} \frac{1}{N}\log \P_{N,\ell N}^\ssup{\beta}\Bigl( \Bigl| \frac{M_N}{N} -q_{\beta,\ell}\Bigr|\geq q_{\beta,\ell} \delta^\ssup{1}_\beta \Bigr) &\leq - q_{\beta,\ell} \delta^\ssup{2}_\beta \\
			\limsup_{N\to \infty} \frac{1}{N}\log \P_{N,\ell N}^\ssup{\beta}\Bigl( || \nu_N - \mathrm{Geom}( \tfrac{q_{\beta,\ell}}{1+q_{\beta,\ell}})||_\mathrm{TV} \geq   \delta^\ssup{1}_\beta \Bigr) &\leq  - q_{\beta,\ell} \delta^\ssup{2}_\beta\\
		\limsup_{N\to \infty} \frac{1}{N}\log \P_{N,\ell N}^\ssup{\beta}\Bigl( ||\widehat \nu_N  - \mathrm{Exp}(\beta p(\beta,\ell)) ||_\mathrm{TV} \geq   \delta^\ssup{1}_\beta \Bigr)
		&  \leq - q_{\beta,\ell} \delta^\ssup{2}_\beta. 
	\end{align*} 
\end{theorem}
\noindent The theorem is proven in Section~\ref{sec:canoproof2}. 

These estimates imply three laws of large numbers for $N\to\infty$ under the distribution $\P_{N,\ell N}^\ssup{\beta}$ at sufficiently low temperature with exponentially fast decay of the deviation from the mean by some threshold that is vanishingly small when $\beta$ is large. In particular, the number of clusters, $M_N$, behaves like $Nq_{\beta,\ell}(1+O( \delta^\ssup{1}_\beta))$ with a probability converging to $1$ exponentially fast. Furthermore, the number of clusters of size $k$ behaves like $N q_{\beta,\ell} (1-q_{\beta,\ell})^{k-1}(1+O( \delta^\ssup{1}_\beta))$ for every $k\in\N$, and hence the average cluster cardinality is about $1 / q_{\beta,\ell}$. Moreover, the distribution of a typical crack length (distance between neighboring clusters) is approximately an exponential variable with parameter $\beta p(\beta,\ell)$ and hence on average of size $1/\beta p(\beta,\ell)$. 

Theorem~\ref{thm:lowdens} makes no statement about the spacings inside the clusters (however, see \cite{jkst19} for more precise assertions in the constant-pressure model), but Lemma~\ref{lem:cluster-lengths}(a) implies that the average spacing is $\approx a$. Hence, the $N$ particles in the interval $[0,\ell N]$ are, with high probability, organized into $Nq_{\beta,\ell}$ clusters that cover each an interval of length $a / q_{\beta,\ell}$ and the same number of gaps in between, each with a size $\approx 1/(\beta p(\beta,\ell))$. Since $a +  q_{\beta,\ell}/(\beta p(\beta,\ell))\approx \ell$, which follows from a comparison of Theorem~\ref{thm:fe}(b) with \eqref{eq:qldef}, this explains how the $N$ particles fill up the system of length $\ell N$. 

We finally remark that our asymptotic estimates are essentially independent of the choice of $R$. To leading order this is a consequence of Eq.~\eqref{e0betalimits}. It also follows a posteriori from Theorems~\ref{thm:fe} and~\ref{thm:lowdens} as the crack length, for a fixed $R$, is exponentially distributed with parameter $\beta p(\beta,\ell)=\e^{- \beta e_\mathrm{surf}^\ssup{R}(\beta)/2+o(\beta)}$ which is itself exponentially small in $1/\beta$. For any other $R' > R$ the probablity of finding spacings which are larger than $R$ but not larger than $R'$ thus becomes negligible at large $\beta$ exponentially fast in $1/\beta$. 

\begin{remark}
In the elementary case of nearest neighbor models~\cite{takahashi42} (i.e., $m=1$) and smooth $v$, one has $a= \mathrm{argmin}\, v(r)$, $e_\mathrm{surf}^\ssup{R}(\beta) = - e_0^\ssup{R}(\beta)$ and  
\begin{align*}
	 e_0^\ssup{R}(\beta) &= -\frac{1}{\beta}\log \Bigl(\int_0^R \e^{-\beta v(r)}\dd r\Bigr) 
		= e_0 + \frac{1}{\beta}\log \sqrt{2\pi \beta v''(a)} + O(\beta^{-3/2}),\\
    \beta p(\beta,\ell) &= (1+o(1))\frac{\exp( -\beta e_\mathrm{surf})}{\sqrt{2\pi \beta v''(a)(\ell - a)}}.
\end{align*}
In particular, the $R$-dependence is explicitly seen to enter in exponentially small correction terms only. 

Harmonic approximations in case of more general pair potentials $v$ would require to replace $v''(a)$ by more complicated terms from Hessians or WKB expansions~\cite{helffer-book, shapeev-luskin17}, see also~\cite[Section 2.3]{jkst19}. (For related techniques in the context of computational approximation schemes for the simulation of atomistic materials see \cite{blanc-lebris-legoll-patz10,binder-luskin-perez-voter2015,  shapeev-luskin17,braun-duong-ortner2020}.) We do not pursue this here.
\end{remark}

\section{Proof ingredients and heuristics} \label{sec:heuristics} 

One-dimensional systems are best treated in the {\em constant-pressure ensemble}, also called {\em isothermal-isobaric} or {\em NpT ensemble}, which does not fix the length of the $N$-particle chain but instead fixes the external pressure. We formulate and prove all the results analogous to Theorems~\ref{thm:fe} and \ref{thm:lowdens} for the constant-pressure ensemble in Section~\ref{sec:pressure} (Theorems~\ref{thm:ge} and \ref{thm:gibbsfe}) and derive Theorems~\ref{thm:fe} and \ref{thm:lowdens} from them in Section~\ref{sec:canonical}. 

In the present section, we introduce the constant-pressure ensemble in Section~\ref{sec-conspress} and give in Section~\ref{sec:heuristics2} extensive heuristics about what properties are to be expected and how the various quantities behave and how they are related to each other. We also introduce and explain the effective model to which we will compare the ensemble when we carry out the proofs in Sections~\ref{sec:aux}--\ref{sec:pressure}. In Section~\ref{sec:compete} we give a modification of the heuristics in a case that we are not considering rigorously in the present article; it leads to a slightly different picture.

\subsection{Equivalence of ensembles and pressure-density (stress-strain) relation}\label{sec-conspress}

The partition function of the {\em constant-pressure ensemble} at pressure $p$ is defined as
\begin{equation}\label{QNdef}
\begin{aligned}
	Q_N(\beta,p) &= \int_0^\infty \e^{-\beta p L} Z_N(\beta,L)\dd L \\
	&= \int_{\R_+^{N-1}} \e^{-\beta [U_N(z_1,\ldots,z_{N-1}) + p \sum_{j=1}^{N-1} z_j ]} \dd z_1\cdots \dd z_{N-1}. 
\end{aligned}
\end{equation}
We write $\Q_N^\ssup{\beta,p}$ for the corresponding probability measure on $\R_+^{N-1}$ with probability density $\vect{z}=(z_1,\dots,z_{N-1})\mapsto Q_N(\beta,p)^{-1} \exp(- \beta [U_N(\vect{z})+ p\sum_{j=1}^{N-1} z_j])$. The Gibbs free energy (also called free enthalpy) per particle is 
\bes
	g(\beta,p) = - \lim_{N\to \infty}\frac{1}{\beta N}\log Q_N(\beta,p). 
\ees
The existence of the limit is well-known, moreover $p\mapsto g(\beta,p)$ is concave and it is related to the Helmholtz free energy by the relations~\cite[Chapter~5.6.6]{ruelle-book69}
\be \label{eq:gfrel}
	g(\beta,p) = \inf_{\ell>0} \bigl( f(\beta,\ell) + p\ell\bigr)\qquad\mbox{and}\qquad f(\beta,\ell) = \sup_{p>0} \bigl( g(\beta,p) - p\ell \bigr),
\ee
which formulate the equivalence of the ensembles at the level of thermodynamic potentials. By standard results on Legendre transforms, as $f(\beta,\cdot)$ is strictly convex and continuously differentiable, $g(\beta,\cdot)$ is strictly concave and continuously differentiable, moreover 
\be \label{eq:pflgrel}
	p = - \frac{\partial f}{\partial \ell}(\beta,\ell) \qquad \Longleftrightarrow\qquad \ell = \frac{\partial g}{\partial p}(\beta,p). 
\ee	
Explicit computations on the equivalence of ensembles and the stress-strain (or force-elongation) relation for one-dimensional systems with nearest or next-nearest neighbor interactions, in a context closer to applications to materials modelling, are given by Legoll and Leli{\`e}vre \cite[Section 2]{legoll-lelievre2012}, see also \cite{blanc-lebris-legoll-patz10}.

\subsection{Effective gas of defects} \label{sec:heuristics2}

An important quantity is the {\em truncated constant-pressure partition function}
\begin{equation}\label{truncpartfct}
	Q_k^\ssup{R} (\beta,p):= \int_{[0,R]^{k-1}} \e^{-\beta [U_k(\vect{z}) +p \sum_{j=1}^{k-1} z_j]} \dd \vect{z}, 
\end{equation}
which restricts to small gaps and describes a cluster of cardinality $k$. Let us give heuristics about its behavior for large $\beta$ and how it is used for a description of the entire constant-pressure ensemble in terms of a decomposition in its clusters and the gaps in between. We assume that $\beta\to \infty$ and $\beta p_\beta\to 0$.

We look at a realization of the $N$-particle ensemble with $n\in \{0,\ldots,N-1\}$ cracks. For $0 = i_0<i_1<\cdots <i_{n+1} = N$, let 
\begin{equation}\label{Bndef}
 B_N(\vect{i})=B_N(i_1,\dots,i_n)=\big\{\vect{z}=(z_1,\dots,z_{N-1})\in \R_+^{N-1}\colon z_j\geq R\Longleftrightarrow j\in\{i_1,\dots,i_n\}\big\}
\end{equation}
be the collection of configurations (chains) that have large gaps (the cracks) precisely at the places $i_1,\dots,i_n$. Suppose that interaction across cracks can be neglected. It is plausible that $Q_k^\ssup{R} (\beta,p_\beta) \approx \e^{-\beta E_k}$. Then, on the event $B_N(\vect{i})$, the entire chain  decomposes into $n$ cracks and $n+1$ clusters:
\bes
	\int_{B_N(\vect{i})} \e^{-\beta [U_N(\vect{z}) +p_\beta \sum_{j=1}^{N-1} z_j]} \dd \vect{z} \approx \Bigl(\int_R^\infty \e^{-\beta p_\beta r} \dd r\Bigr)^n\, \prod_{k=1}^{n+1} \e^{-\beta E_{i_k - i_{k-1}}}.
\ees
Set $V(k) = E_k - k e_0 - e_\mathrm{surf}$. Notice that
\be \label{eq:interact}
  \sum_{k=1}^{n+1} E_{i_k - i_{k-1}} = N e_0 + (n+1) e_\mathrm{surf} + \sum_{k=1}^{n+1} V(i_k- i_{k-1}).
\ee
Thus setting 
\[ %\label{qdef}
	q = q_\beta = \frac{\exp(-\beta [e_\mathrm{surf} + p_\beta R])}{\beta p_\beta}
\approx \frac{\exp(-\beta e_\mathrm{surf})}{\beta p_\beta},
\]
we get
\be
	Q_N(\beta,p_\beta)  
	\approx \e^{-\beta (N e_0 + e_\mathrm{surf})} \sum_{n=0}^{N-1}q^n \sum_{1\leq i_1<\cdots < i_n\leq N-1}  \e^{-\beta \sum_{k=1}^{n+1} V(i_k-i_{k-1})}.\label{eq:qeff}
\ee 
We recognize the partition function for an {\em effective lattice gas} on $\{1,2,\dots,N-1\}$ with activity $q$ and interaction potential  $(i,j)\mapsto V(j-i)$. Each site $j=1,\ldots,N-1$ corresponds to a bond $z_j$ between neighboring particles, and a defect is present at $j$ if $z_j\geq R$ is a crack. If $V$ was neglected, then the lattice gas would be {\em ideal}, and the right hand side of \eqref{eq:qeff} would be equal to $ \e^{ - \beta (N e_0+e_\mathrm{surf})}(1+q)^{N-1}$. 

\begin{remark}  \label{rem:defect-eof}
The reader may also think of $- \frac1\beta\log q$ as the Gibbs free energy of formation of a defect.  Computing free energies of defect formation is a non-trivial task, see e.g.~\cite{braun-duong-ortner2020} and the references therein. The Gibbs free energy of defect formation is a sum of two contributions: an energetic contribution $e_\mathrm{surf}$ that accounts for missing interactions across the crack, and an entropic contribution $\frac1\beta \log (\beta p_\beta)$ that comes from integrating over different possible lengths of the crack $z_j\geq R$. At fixed pressure only the energetic contribution would survive in the zero-temperature limit, however in our context the  pressure is exponentially small in $\beta$ (see Eq.~\eqref{eq:pqchoice} below) and both the energetic and the entropic contributions are relevant. 
\end{remark} 

From the definition \eqref{eq:surface} of $e_0$ and $e_\mathrm{surf}$ we know that $V(k)\to 0$ as $k\to \infty$. Hence, we work in a perturbative regime and need to control that $V$ is small enough in an appropriate sense. Criteria for this are well-known. Indeed, according to \cite[Theorem 4.2.3]{ruelle-book69}, if the quantity $q C(\beta)$ is small, where
\[ %\label{Cbetadef}
C(\beta)  = \sum_{k=1}^\infty |\e^{-\beta V(k)} - 1|,
\]
then the effect of interactions is negligible, and we may approximate the effective model by the ideal lattice gas. 

In this approximation,  under the assumption that $q C(\beta)$ is small, we get a number of consequent crucial approximations. Indeed, the collection of bonds (effective lattice sites) is approximately independent, and the probability that site $j$ is occupied ($z_j\geq R$) approaches the (tiny) number $q/(1+q)$ (and with the remaining probability $1/(1+q)$, it is not). As a consequence, the number of particles in successive clusters becomes geometric with this parameter. Furthermore, the length $z_j-R$ of a crack minus $R$ is approximately exponentially distributed with small parameter $\beta p_\beta$ and expected length $1/\beta p_\beta$. Additional arguments that analyse the energy term show that any length of a spacing inside a cluster approaches the ground state spacing $a$. As a consequence, any spacing is $\approx a$ with probability $1/(1+q)$ and $\approx R+1/\beta p_\beta$ otherwise. In particular, the average length of a spacing then is 
\[ %\label{eq:lbehav} 
	\ell \approx \frac{1}{1+q}\times a + \frac{q}{1+q} \Bigl( R+ \frac{1}{\beta p_\beta}\Bigr) .
\]
Assuming that $q$ is very small (low density of defects because of large $\beta$), 
the only way that a length $\ell\in(a,\infty)$ can be achieved is that 
$$ 
\frac{q}{\beta p_\beta}\approx \frac{\exp(-\beta e_\mathrm{surf})}{(\beta p_\beta)^2}\to \ell - a,
$$
which yields 
\be \label{eq:pqchoice}
	\beta p_\beta\approx \frac{\exp(-\beta e_\mathrm{surf}/2)}{\sqrt{\ell - a}}\qquad\mbox{and}\qquad q\approx \sqrt{\ell - a}\,  \e^{-\beta e_\mathrm{surf}/2}.  
\ee 
Hence, the smallness of $q C(\beta)$ would lead to a complete picture of the behavior of the chain, which is the one that we describe in Theorems~\ref{thm:fe} and \ref{thm:lowdens}.

The number $C(\beta)$ is a common measure in statistical mechanics for the overall strength of the interactions, see~\cite[Definition 4.1.2]{ruelle-book69}. However, there is {\em a priori} no reason that it be small. In general, it can go to infinity exponentially fast as $\beta\to\infty$. Indeed,  
\bes
	\liminf_{\beta \to \infty} \frac{1}{\beta} \log C(\beta) \geq - \inf_{k\in \N} V(k) = - \inf_{k\in \N} (E_k - k e_0 - e_\mathrm{surf}) \geq  e_0 + e_\mathrm{surf}
\ees
(recall $E_1 =0$). Under our assumptions on the pair potential, we have $E_k \geq (k-1)e_0$ for all $k\in \N$~\cite[Lemma 3.2]{jkst19} and hence $e_\mathrm{surf} + e_0 \geq 0$. (In particular, $V(k) \geq 0$, which justifies the first inequality in the above estimate.) As soon as the inequality is strict, we find that $C(\beta) \to \infty$ exponentially fast. 

Hence, our plan works only if $q=q_\beta$ vanishes as $\beta\to\infty$ fast enough. If  the pressure $p_\beta$ goes to zero not too fast so that $\beta p_\beta \gg \exp(-\beta e_\mathrm{surf}) \to 0$---for example, by choosing $p_\beta$ as in~\eqref{eq:pqchoice}---we see that $q=q_\beta \to 0$. A necessary condition for $q_\beta C(\beta)\to 0$, when $q_\beta$ is as in~\eqref{eq:pqchoice}, is certainly that $e_0 + e_\mathrm{surf}/2\leq 0$. This is indeed the case in which we  are working in the present article, see Lemma~\ref{lem:surfbound} and Theorem~\ref{thm:codec}.

\subsection{An alternative scenario} \label{sec:compete}

Let us present a modified heuristics in the case where $e_\mathrm{surf}/2> |e_0|$, which we do not  handle rigorously in this article. We still assume that $\inf_{k\in \N} (E_k - k e_0) = E_1 - e_0 = |e_0|$, as is proved in~\cite[Lemma 3.2]{jkst19}.

Let us make one more approximation step on the right-hand side of \eqref{eq:qeff}. We introduce the solution $u=u_\beta$ of
\be \label{eq:uheu}  
	q_\beta\sum_{k=1}^\infty  u_\beta^k   \,\e^{-\beta V(k)}= 1,
\ee
and introduce an $\N$-valued random variable $T_\beta$ which assumes the value $k\in\N$ with probability $q_\beta u_\beta^k \e^{- \beta V(k)}$.  Then independent copies of $T_\beta$ play the role of the cardinalities of the clusters. (Notice that the geometric distribution from Theorem~\ref{thm:lowdens} is recovered with the approximation $V(k) =0$, under which $u_\beta = 1/(1+q_\beta)$.)  The right-hand side of \eqref{eq:qeff} can be further transformed using these variables, which we carry out in Section~\ref{sec:aux}.

With the ansatz  $u_\beta\approx \exp(- \beta t_\beta)$ where $t_\beta = \exp(- \beta e_\mathrm{surf}/2)$, and with the help of \eqref{eq:pqchoice}, Eq.~\eqref{eq:uheu} becomes
\be \label{eq:heu2}
	\sum_{k=1}^\infty \e^{- k \beta t_\beta} \e^{-\beta (E_k - k e_0)} \approx \beta p_\beta.
\ee
Further approximations yield (splitting the sum at $k=1$) 
\bes
	\beta p_\beta \approx \e^{- \beta t_\beta}\e^{-\beta (E_1 -  e_0)}+\sum_{k=2}^\infty \e^{- k \beta t_\beta} \e^{-\beta (E_k - k e_0)} 
	\approx \e^{- \beta  |e_0|}  + \frac{\e^{- 2 \beta t_\beta} \e^{-\beta e_\mathrm{surf}}}{1- \exp( - \beta t_\beta)} \approx \e^{-\beta |e_0|} + \frac{\e^{-\beta e_\mathrm{surf}}}{\beta t_\beta}.
\ees
Hence 
\begin{equation}\label{heuristicsstop}
	\lim_{\beta \to \infty} \frac{1}{\beta} \log (\beta p_\beta) = - \lim_{\beta \to \infty} \min \Big\{|e_0|, e_\mathrm{surf} +\frac{1}{\beta} \log  (\beta t_\beta)\Big\} = -\min\{|e_0|, \smfrac 12 e_\mathrm{surf}\}=- |e_0|,
\end{equation}
where we used in the last step that $|e_0| < \frac 12 e_\mathrm{surf}$. (Here is the point at which the heuristics deviates from the situation considered in this article.) In order to find the expectation of $T_\beta$, we approximate, again splitting the sum at $k=1$,
\bes 
	\sum_{k=1}^\infty k \e^{- k \beta t_\beta} \e^{-\beta (E_k -k e_0)}
		\approx \e^{-\beta |e_0|} + \frac{\exp(-\beta e_\mathrm{surf})}{(\beta t_\beta)^2} \sim \frac{\exp(-\beta e_\mathrm{surf})}{(\beta t_\beta)^2} = \e^{o(\beta)}.
\ees
Hence, using that $T_\beta$ assumes each $k\in\N$ with probability $q u_\beta^k \e^{- \beta V(k)}$ and recalling that $V(k)=E_k-k e_0-e_{\rm surf}$, we see that the average cardinality of a given cluster is 
\bes	 
	\E[T_\beta] \approx \frac{1}{\beta p_\beta} \frac{\exp(- \beta e_\mathrm{surf})}{(\beta t_\beta)^2} \approx \e^{\beta |e_0| +o(\beta)}.
\ees
Accordingly, the average number of clusters is $\approx N/\E[T_\beta] \approx N \exp( - \beta|e_0|)$. We expect the chain of atoms to have a length given by the number of clusters times the sum of average cluster length and average crack length
\bes
	\frac{N}{\E[T_\beta]} \Big( a \E[T_\beta] + R + \frac{1}{\beta p_\beta}\Big)
		\approx N \Bigl( a + \frac{(\beta t_\beta)^2}{\exp(-\beta e_\mathrm{surf})} \Bigr).
\ees
Since our container has length $N\ell$, this suggests $\beta t_\beta \approx \sqrt{\ell -a } \exp( - \beta e_\mathrm{surf}/2)$, in agreement with our ansatz for $t_\beta$.  This leads us altogether to a picture that is slightly different from Theorem~\ref{thm:lowdens}: 
\begin{itemize}
	\item The pressure is $\beta p_\beta \approx \e^{-\beta |e_0|}$ instead of $\approx \e^{-\beta e_{\rm surf}/2}$. 
	\item The fraction of defects is $\approx\beta p_\beta \sqrt{\ell-a}$. 
	\item The cluster size $T_\beta$ is no longer approximately geometric anymore because the dominant contribution to the infinite sum~\eqref{eq:heu2} comes from bounded $k$. Put differently, defects tend to gather at finite mutual distance.
\end{itemize}
On the other side, the following features are the same in both heuristics:
\begin{itemize}
	\item $\E[T_\beta]\to \infty$, and the \emph{size-biased} law (the cardinality of the cluster containing a given particle) $\P(\widetilde T_\beta=k) = k \P(T_\beta =k)/ \E[T_\beta]$ is still comparable to a size-biased geometric law with parameter $\exp(- \beta t_\beta)$. 
	\item The crack length has an exponential law with parameter $\beta p_\beta$ and hence an average length $\approx 1/\beta p_\beta$ (but with value $\approx \exp( \beta |e_0|)$, see above).  
\end{itemize}

This heuristics provide intuition also in less restrictive situations than under our precise assumptions of Section~\ref{sec:ass}. In particular, Eq.~\eqref{eq:heu2} is applicable with $e_\mathrm{surf} = \liminf_{k\to \infty} (E_k - k e_0)$ when $E_k-ke_0$ is not convergent, which can happen for non-convex interactions where parity plays a role~\cite{braides-cicalese07}. The comparison of $e_\mathrm{surf}$ and $e_0$ as well as the evaluation of $\inf_{k\in \N} (E_k - k e_0)$, which in general need not be equal to $|e_0|$, are in turn closely related to the location of fracture in zero-temperature models~\cite{braides-cicalese07, scardia-schloem-zanini11}.

\section{Weakly interacting lattice gas} \label{sec:aux}

In this section we analyze an abstract lattice gas model motivated by Eq.~\eqref{eq:qeff}. For the reader's orientation, it is helpful to recall the heuristics of Sections~\ref{sec:heuristics2} and \ref{sec:compete} until \eqref{heuristicsstop}. We will have no parameter $\beta$.

In Section~\ref{sec:effective} we introduce the model and find some first properties of its free energy in terms of standard renewal theory. In Section~\ref{sec:geometric} we introduce the random variable $T$ that plays the role of the number of clusters and derive precise estimates about its distance to the geometric distribution. Large-deviation principles and the relevant estimates are derived in Section~\ref{sec:LD}.

\subsection{Effective free energy} \label{sec:effective}

It is convenient to work with $f(k) = \e^{-\beta V(k)} - 1$ rather than the interaction itself.  Thus we assume that a number $q>0$ and a map $f\colon \N\to [-1,\infty)$ are given such that 
\be \label{ass:f}
	\eps:=q\sum_{k=1}^\infty |f(k)|<1.
\ee
Put differently, we assume that $C:=\sum_{k=1}^\infty |f(k)|$ is finite, and $q= \eps/C$ for $0 \le \eps < 1$. We think of $q$ and $\eps$ as small numbers, whereas $f(k)$ can be large, at least for small $k$.  Consider the partition function
\be \label{eq:ZN}
	\mathcal{Z}_N(q): 
	= \sum_{n=1}^{N} q^{n} \sum_{0=i_0<\cdots <i_{n} = N} \prod_{k=1}^{n} \bigl( 1+ f(i_k - i_{k-1})\bigr). 
\ee
It can be studied either directly, using standard tools of statistical mechanics such as cluster expansions, or with the help of standard renewal theory from probability theory; see \cite[Chapter XI]{feller-vol2}. We are going to use the latter. Let  $u  \in (0,1)$ be the unique solution of 
\be \label{eq:renewalsolv}
	q \sum_{k=1}^\infty \bigl( 1+ f(k)\bigr) u^k = 1. 
\ee
and let $T, T_1, T_2,T_3,\dots$ be independent identically distributed random variables with law 
\be \label{eq:iid} 
	\P(T = k) = q \bigl( 1+ f(k)\bigr) u^k, \qquad k\in \N.
\ee
(Then $T_i$ plays the role of the cardinality of the $i$-th cluster.) The partition function~\eqref{eq:ZN} of the defect gas is related to the random variables $T_i$ by  
\begin{align*} 
	\mathcal{Z}_N(q)
	& = u^{-N} \sum_{n=1}^{N} \sum_{0= i_0<\cdots <i_{n} = N} \P(T_1=i_1-i_0,\ldots,T_{n} = i_{n} - i_{n-1}). \\
	& = u^{-N}\P(\exists n\in \N\colon T_1+\cdots + T_{n} = N).
\end{align*}
(Recall that the integers $i_k$ correspond to locations of cracks and the variables $T_k$ count cluster sizes, i.e.,  the number of points enclosed between two successive cracks.)
It follows from standard renewal theory that $u^N\mathcal{Z}_N(q) \to 1/\E[T]$ as $N\to \infty$, hence the effective free energy is given by
\be \label{eq:fe}
 \lim_{N\to \infty} \frac{1}{N} \log \mathcal{Z}_N(q) = - \log u.
\ee
It is actually close to $\log (1+q)$, as we have that
\begin{equation}\label{prop:fe}
	|1-(1+q) u| \leq \frac{q \eps}{1-\eps}.
\end{equation}
Indeed, by a straightforward computation, Eq.~\eqref{eq:renewalsolv} is equivalent to 
		\[ %\label{eq:difference}
			1 - (1+q) u = q u \, \frac{ q \sum_{k=1}^\infty f(k) u^k }{1 - q \sum_{k=1}^\infty  f(k) u^k}. 
		\]
		Then \eqref{prop:fe} follows from the fact that $u\in (0,1)$ and the monotonicity of $x\mapsto x/(1-x)$ in $(-1,1)$.

In addition to the formula~\eqref{eq:fe} for the free energy, renewal theory also yields an explicit description of the thermodynamic limit: as $N\to \infty$, the bulk behavior is given by a stationary renewal process. In our setup, this means in particular that the probability that a nearest neighbor bond $(i,i+1)$ is broken has probability $1 / \mu$ where $\mu = \E[T]$, and given that the bond is broken, the particle $i+1$ belongs to a $k$-cluster with probability $\P(T=k)$ (same statement for particle $i$).

\subsection{Approximately geometric variables} \label{sec:geometric}

We continue our analysis of the abstract gas model introduced in Section~\ref{sec:effective}. For non-interacting defects, that is, $V(k)=0$, i.e., $f\equiv 0$, we have $u=1/(1+q)$ and the random variable $T$ (standing for the cluster size) has precisely a geometric law. For weak interactions and small $q$, we may expect approximately a geometric law. In this section we provide explicit estimates. 

\begin{lemma} \label{lem:ibound} 
	For every $r\in\N$ and $\tau>0$, 
	$$q  \sum_{k=1}^\infty k^r |f(k)| \e^{-k\tau q} \leq \frac{r!\,\eps }{(\tau q)^r}.$$
\end{lemma} 

\begin{proof} 
	For $t\in \C$ with $\Re t\leq 0$, let $g(t):= q \sum_{k=1}^\infty |f(k)| \e^{t k}$. The function $g$ is analytic on the open half-plane $\{t\in\C\colon \Re t<0\}$ and bounded by $\eps$ on the closed half-plane $\{t\in\C\colon \Re t\leq 0\}$. Let $t= - \tau q \in (-\infty,0)$ and $r\in \N$. By Cauchy's formula, 
	\bes 
		g^{\ssup r}(t) = \frac{r!}{2\pi\mathrm{i}} \oint_{|z-t|=\tau q} \frac{g(z)}{(z-t)^{r+1}} \dd z
	\ees
	hence $\sum_{k=1}^\infty q |f(k)|k^r \e^{kt}  = g^{\ssup r}(t) \leq  r! \eps/ (\tau q)^r$.
\end{proof} 

Let $G$ be a geometric random variable with law $\P(G=k) = q (1+q)^{-k}$ for $k\in \N$. We  compare the laws $\mathcal{L}(T)$ and $\mathcal{L}(G)$ of $T$ and $G$.
Let $\widetilde T$ and $\widetilde G$ be the size-biased variables associated with $T$ and $G$, i.e., $\P(\widetilde T= k) = k \P(T=k)/\E[T]$ and $\P(\widetilde G=k)= k q^2 (1+q)^{-k-1}$. Recall the total variation norm $\|\mu\|_{\rm TV}=\sum_{k\in\N}|\mu(k)|$ of a signed measure $\mu$ on $\N$.

\begin{lemma} \label{lem:geom1}
	As $\eps,q\to 0$,
	$$ u  = \frac{1}{1+q} \bigl(1+O(q \eps)\bigr) \qquad \mbox{and}\qquad
		\E[T]  =   \frac{1}{q}\bigl(1 + q +O(\eps)\bigr).
	$$
	Moreover,
	$$
	\|T-G\|_{\rm TV}\leq O(\eps),\qquad \|\widetilde T-\widetilde G\|_{\rm TV}\leq O(\eps).
	$$
\end{lemma} 

\begin{proof} 
	The estimate on $u$ follows from \eqref{prop:fe} and the assumption~\eqref{ass:f}.
	Next we compute
	\be \label{eq:t1}
		\E[T] = \sum_{k=1}^\infty k \P(T=k)
		 = q \sum_{k=1}^\infty k  f(k) u^k + \frac{q u}{(1- u)^2}. 
	\ee
	Eq.~\eqref{prop:fe} shows that eventually $u\leq\exp(-q/2)$. Consequently, by Lemma~\ref{lem:ibound}, the first term on the right-hand side of 
	Eq.~\eqref{eq:t1} is of order $O(\eps/q)$. For the second term, set $\hat u = (1+q) u = 1+ O(q\eps)$ and note
	\[ %\label{eq:quotient}
		\frac{q u}{(1- u)^2} =\frac {q/(q+1)} {(1- 1/(1+q))^2} \times \hat u \frac{q^2}{(1+ q- \hat u )^2} = \frac{1}{q}\bigl(1 + q + O(\eps)\bigr). 
	\]
	The estimate for $\E[T]$ follows. For the total variation distance, we estimate, using \eqref{eq:iid} and \eqref{ass:f},
	\bes
		\|T-G\|_{\rm TV}=\sum_{k=1}^\infty |\P(T=k) - \P(G=k)| \leq O(\eps) + \sum_{k=1}^\infty \frac{q}{(1+q)^k} \bigl| \hat u^k-1\bigr|. 
	\ees	
   The latter term is equal to 
   \bes
\Bigl| \sum_{k=1}^\infty \frac{q}{(1+q)^k} (1 - \hat u^k) \Bigr|= \Bigl| 1 - \frac{q \hat u}{1+q} \times \frac{1}{1- \hat u/[1+q]} \Bigr| 
		= \Bigl| 1 - \frac{q\hat u}{q+ 1- \hat u} \Bigr| = O(\eps).
	\ees
	It follows that $\|T-G\|_{\rm TV}$ is of order $O(\eps)$.  The size-biased distributions are treated in a similar way.
\end{proof} 

\noindent We also need some control of the cumulant generating function of $T$ and its Legendre transform. Let 
\be \label{eq:phidef}
	\varphi(t) = \log \E[\e^{tT}] = \log \Bigl(q\sum_{k=1}^\infty (1+ f(k))u^k \e^{tk}\Bigr),\qquad I(x) = \varphi^*(x) = \sup_{t \in \R} (t x - \varphi(t) ).
\ee
We have $\varphi(t) = \infty$ for $t \geq - \log u = \log(1+q +o(q))$. The function $\varphi$ is a smooth, increasing, strictly convex bijection from $(- \infty, -\log u)$ onto $\R$. As is well-known, $\varphi'(0) = \E[T]=:\mu$, $\varphi''(0) =\var(T)$ (the variance of $T$), $I(\mu)=0$ and $I''(\mu) = 1/\var[T]$. In view of the geometric approximation, we expect $\var[T]\approx 1/q^2$, and that the quadratic approximation to $I(x)$ for $x\approx \mu$ becomes $I(x) \approx \frac{1}{2}q^2(x-\mu)^2$. The next lemma provides a corresponding lower bound with some uniformity as $q,\eps \to 0$. 

\begin{lemma}\label{lem:ilevel} 
	Let $\mu =\E[T]= (1 + q + O(\eps))q^{-1}$. Then there exist $c,\delta>0$ such if $q,\eps\in [0,\delta]$, then, for all $x\in \R$ with $|x-\mu|\leq \delta/q$,
	\bes
		I(x) \geq \frac{1}{2 c} q^2 (x-\mu)^2.
	\ees
\end{lemma} 

\begin{proof} 
	Fix $\tau \in (0,1)$. Then for some $c_\tau, \delta_\tau>0$, and all $q\leq \delta_\tau$, 
	\be \label{eq:mocubo}
		\sup_{|t|\leq \tau q}\E[ T^2 \e^{t T}] \leq \frac{c_\tau}{q^2}. 
	\ee
	Indeed, for $|t| \leq \tau q$ and abbreviating $w= u \e^{t}$, we have
	\begin{align} 
		\E[ T^2 \e^{t T}]  & = q \sum_{k=1}^\infty k^2(1+ f(k)) w^k %= q (w^2\sum_{k=1}^\infty k (k-1) w^{k-2} + w \sum_{k=1}^\infty k w^{k-1}) + \sum_{k=1}^\infty q f(k) w^k 
		= \frac{2 q w^2}{(1-w)^3} + \frac{q w}{(1-w)^2}+q\sum_{k=1}^\infty  k^2 f(k) w^k. \label{eq:mocubob} 
	\end{align} 
	Notice $w = (1 + q)^{-1} (1 + O(\eps q)) \e^t = (1 - q + O(q^2 + \eps q))(1 + \tau q + O(q^2)) = 1 - (1 - \tau) q + O(q^2 + \eps q)$. In particular, $w\to 1$. Hence,	
	choosing $\delta= \delta_\tau>0$ small enough, we find that for all $\eps, q\leq \delta$, we have $w \leq \exp( - (1-\tau)q/2 )$ and $1 - w \geq (1-\tau)q/2$. The bound~\eqref{eq:mocubo} now follows from~\eqref{eq:mocubob} and Lemma~\ref{lem:ibound}. Noting that
	\bes 
		\bigl|\E[\e^{t T}] - 1 - t\, \E[T] \bigr| \leq \frac{t^2}{2} \E[T^2 \e^{|t| T}], 
	\ees
	and recalling $\mu = \E[T]\sim 1/q$ by Lemma~\ref{lem:geom1}, we deduce (note that $\log (1+u)\leq u$) 
	\[ %\label{eq:mocuboc}
	 \varphi(t)\leq t \mu + \frac{c_\tau}{2 q^2} t^2
	\]
	for $|t|\leq \tau q$ and $\eps, q\leq \delta_\tau$.  It follows that 
	\bes
		I(x) \geq \sup_{|t|\leq \tau q} \bigl( t x - t\mu - \frac{c_\tau}{2 q^2} t^2 \bigr) =\frac{q^2}{2 c_\tau}  \min \big\{ (x- \mu)^2, (\tau c_\tau / q)^2 \big\}. 
	\ees
	If $\delta$ is chosen small enough, then indeed $\min \{ (x- \mu)^2, (\tau c_\tau / q )^2 \} = (x-\mu)^2$ for $|x-\mu|\leq \delta/q$. 
	%Going back to~\eqref{eq:mocubob}, we see that the constant $c_\tau$, for small $\tau$, is of the order of $ (1-\tau)^{-3} + (1-\tau)^{-2}+\eps$, which for small enough $\tau$ and $\eps$ is smaller than $4$. 
	\end{proof}

\subsection{Large deviations}\label{sec:LD}

The system that we wish to investigate can be  expressed exactly in terms of a lattice gas of defects as in Section~\ref{sec:effective} only when interactions across cracks vanish, i.e., for compactly supported potentials, see Assumption 4 in Section~\ref{sec:ass}. In the general case, we estimate the contribution of interactions across cracks by some small number $\lambda$ times the number of cracks, see Lemma~\ref{lem:rep} below. In order to quantify the effect of this small contribution we use large deviations theory. Providing this is the purpose of the present section. We keep all the notation from Sections~\ref{sec:effective} and \ref{sec:geometric}.

For the reader's convenience, we briefly repeat what a large deviations principle (LDP) is, see \cite{dembo-zeitouni} for more about this theory. We say that a sequence of random variables $X_N$ with values in a Polish space $\mathcal X$ satisfies an LDP with speed $N$ and with lower semi-continuous rate function $I\colon\mathcal X\to[0,\infty]$ if for every open set $G\subset \mathcal X$ and every closed subset $F\subset\mathcal X$,
$$
\limsup_{N\to\infty}\frac 1N\log \P(X_N\in F)\leq - \inf_F I,\qquad \liminf_{N\to\infty}\frac 1N\log \P(X_N\in G)\geq - \inf_G I.
$$
The intuitive idea behind this is that $\P(X_N\approx x)\approx \e^{-N I(x)}$ for $x\in\mathcal X$. Below, we will be working with $\mathcal X$ chosen as $\N$ and the set of probability measures on $\N$ and the product of the two.

Let $(M_N)_{N\in \N}$ be a sequence of $\N$-valued random variables with law 
\be\label{eq:mnk}
	\P( M_N = k) = \frac{\P(T_1+\cdots + T_k = N)}{\sum_{n=1}^N \P(T_1+\cdots + T_n = N)}
\ee
for $k=1,\ldots,N$ and $\P(M_N\geq N+1) = 0$. Thus $M_N$ counts the number of renewal intervals between $0$ and $N$ given that there are renewal points at $0$ and $N$. For the chain of atoms, $M_N$ corresponds to the number of clusters (= 1 + number of cracks) in an $N$-particle chain. 
	
\begin{lemma} \label{lem:mnldp}
	 $(M_N/N)_{N\in \N}$ satisfies a large deviations principle with speed $N$ and rate function 
	\bes
		J(y) = \begin{cases}
				y I (y^{-1}), &\quad y > 0,\\
				- \log u, &\quad y = 0, \\
				\infty, & \quad y < 0.
			\end{cases}  
	\ees
	 Moreover for all $\lambda \in \R$, 
	\be \label{eq:gel}
		\lim_{N\to \infty}\frac{1}{N}\log \E\bigl[\e^{\lambda M_N}\bigr] = - \varphi^{-1}(- \lambda).
	\ee
\end{lemma}

\begin{proof} 
	We already know from renewal theory~\cite[Chapter XI]{feller-vol2} that the denominator in Eq.~\eqref{eq:mnk} converges to $1 / \E[T]$. 
	For the numerator, pick  $t< - \log u$ and let  $\hat T,\hat T_1,\hat T_2,\ldots$ be i.i.d. random variables with tilted law $\P(\hat T= k) = \exp( t k) \P(T=k) /\exp(\varphi(t)) $. Then $\E[\hat T] = \varphi'(t)$ and by standard renewal theory 
\bes
	 \sum_{k=1}^N \e^{k \lambda }\P(T_1+\cdots + T_k = N) 
			= \e^{-N t} \sum_{k=1}^N \P(\hat T_1+\cdots+ \hat T_k = N)  =(1+o(1)) \frac{\e^{-N t}}{\E[\hat T]}
\ees
where $\lambda = - \varphi(t)$. It follows that 
\bes
	\lim_{N\to \infty} \frac{1}{N}\log \E\bigl[\e^{\lambda M_N}\bigr] = - t = - \varphi^{-1}(-\lambda)=:\psi(\lambda).
\ees
This proves Eq.~\eqref{eq:gel}.
Now $\varphi$ is a smooth, strictly convex, monotone increasing bijection from $(-\infty,-\log u)$ onto $\R$.  It follows that $\psi\colon \R\to (\log u,\infty)$ is a monotone increasing bijection and strictly convex as well. The G{\"a}rtner-Ellis theorem shows that $(M_N/N)_{N\in \N}$ satisfies a large deviations principle with speed $N$ and rate function $J=\psi^*$. The explicit expression for the Legendre transform $\psi^*$ follows from $I=\varphi^*$ and \cite[Theorem 4]{hiriart-martinez03}.
\end{proof} 

\begin{lemma} \label{lem:tilted}
	There exist $\delta>0$ and $c,C>0$ such that the following holds for all $\eps,q,\lambda\in (0,\delta)$: With 
 $A=\{y\in\R \colon |\mu y- 1 |\geq C \sqrt{\lambda}\}$,
	\bes
		\limsup_{N\to \infty} \frac{1}{N}\log \E\Bigl[\e^{\lambda M_N} \1_{\{M_N/N \in A\}}\Bigr] \leq - c q \lambda. 
	\ees
\end{lemma} 
\noindent Put differently, the dominant contributions to $\E[\exp(\lambda M_N)]$ come from $A^{\rm c }$, i.e., from realizations where $M_N/N \approx q(1\pm O(\sqrt{\lambda})+ O(q+\eps))$.

\begin{proof} 
	By the Cauchy-Schwarz inequality and Lemma~\ref{lem:mnldp}, we have
	\bes %\label{eq:csi} 
		 	\limsup_{N\to \infty} \frac{2}{N}\log \E\Bigl[\e^{\lambda M_N} \1_{\{M_N/N \in A\}}\Bigr] 
		 		\leq - \varphi^{-1}(-2\lambda) - \inf_{y\in A} J(y).
	\ees
	Let us estimate  $- \varphi^{-1}(-2\lambda)$. Let $t=\tau q$ with $\tau \in[-1/4,1/4]$. Recall  $\varphi(0)=0$ and $\varphi'(0)=\mu$. 	
	 Proceeding as in Lemma~\ref{lem:ilevel}, we see that  
	 \bes
		|\varphi(t) - \mu t| \leq \frac{1}{2} t^2 \sup_{|s|\leq q /4} 
		|\varphi''(s)| = O(\tau^2).
	 \ees
	Hence 
	\be \label{eq:vfa}
	\varphi(\tau q) = \mu q \tau + O(\tau^2) = \tau (1 + q + O(\eps)+O(\tau))
	\ee 
	 uniformly in $|\tau|\leq 1/4$. We may thus choose a constant $M \ge 4$ such that in particular  $\varphi(q/M)$ and $\varphi(-q/M)$ are bounded away from zero as $\eps,q\to 0$. Thus we may find $\delta_1>0$ such that if $\eps,q\in [0,\delta_1]$, then $|\varphi(\pm q/M)|\geq 2 \delta_1$. For $|\lambda|\leq \delta_1$, we have $- 2 \lambda = \varphi(\tau q)$ for some $|\tau|\leq 1/M$. We combine with Eq.~\eqref{eq:vfa} and find that $\varphi^{-1} (- 2\lambda) = q \tau = - 2 q \lambda (1+O(q + \eps + \lambda))$. Making $\delta_1$ smaller if necessary, we may assume $\varphi^{-1}(-\lambda)\leq 4 q \lambda$ for $|\lambda|\leq \delta_1$. 
	 
	Next we note that $J(y)\geq 0$ with equality if and only if $y=\mu^{-1}$. Furthermore, because of the strict convexity of $J$,
	\bes
		\inf_{y\in A} J(y) \geq \min \big\{ J(y_-), J(y_+)\big\},\qquad\mbox{where } y_\pm =  \mu^{-1}(1 \pm C\sqrt{\lambda}).
	\ees
	From Lemma~\ref{lem:ilevel}, the definition of $J$ in Lemma~\ref{lem:mnldp}, and the identity $\mu= q^{-1}(1+q+O(\eps))$ we have 
	\bes
		J(y_\pm) \geq \frac{1}{2c}  q^2 y_\pm \bigl( y_\pm^{-1} - \mu\bigr)^2 
			= \frac{1}{2c} q^2 y_\pm^{-1} \mu^2 C ^2  \lambda = \frac{1}{2c} q C^2 \lambda (1+ O(q + \eps + \sqrt{\lambda})). 
	\ees
	For small $\eps,\lambda$ this is larger than, say $q C^2 \lambda/(4c)$. Choosing $C^2 /(4c) > 4$ we find $\inf_A J + \varphi^{-1}(-2 \lambda) \geq (C^2/ 8 - 4) q \lambda$, and the assertion of the lemma follows. 
	\end{proof} 

Next, let 
\be\label{eq:empdist-T}
	\nu_N= \frac{1}{M_N}\sum_{j=1}^{M_N} \delta_{T_j}
\ee
be the empirical distribution of $T_1,\ldots,T_{M_N}$. Then $\nu_N$ is a random variable with values in the space of probability measures on $\N$ equipped with the topology of weak convergence. Note that on this space weak convergence is equivalent to pointwise convergence. 

\begin{lemma}\label{lem:tilted-empirical}
	There exist $\delta>0$ and $c,C>0$ such that the following holds for all $\eps,q,\lambda\in (0,\delta)$. With
	 $B=\{\nu\colon ||\nu - \mathcal{L}(T)||_\mathrm{TV} \geq \sqrt{2\lambda}\}$,
	\bes
		\limsup_{N\to \infty}\frac{1}{N}\log \E\Bigl[\e^{\lambda M_N}\1_{\{\nu_N\in B\}}\Bigr] \leq -  c q \lambda (1+ O(q+\eps)). 
	\ees
\end{lemma} 
Combined with Lemma~\ref{lem:geom1}, we see that dominant contributions to $\E[\e^{\lambda M_N}]$ come from realizations where the total variation distance between the law of $\nu_N$ and law of the geometric variable $G$ defined in Section~\ref{sec:geometric} is of the order of $O(\sqrt{\lambda}) + O(\eps)$. 

\begin{proof}  
	Using Lemma~\ref{lem:mnldp} and Sanov's theorem~\cite[Chapter 6.2]{dembo-zeitouni}, it is not difficult to see 
	that the pair sequence $(M_N/N, \nu_N)_{N\in\N}$ satisfies a joint large deviations principle with speed $N$ and rate function 
	\be\label{eq:ratefcn-cluster}
		\widetilde J(y, \nu) = J(y) + y H\bigl(\nu;\mathcal{L}(T)\bigr)  
					= y \Bigl( I(y^{-1})+ \sum_{k=1}^\infty \nu(\{k\}) \log \frac{\nu(\{k\})}{\P(T=k)}\Bigr) 
	\ee
	for $y > 0$ and $\widetilde J(y, \nu) = \infty$ otherwise. Let $\delta, \lambda, C, A$ be as in Lemma~\ref{lem:tilted}. 	
	Thus $A^{\rm c}= \{y\colon |\mu y - 1| < C \sqrt{\lambda}\}$. We estimate 
	\be \label{eq:tilted-empir-aux}
	\begin{aligned}
		\E\bigl[\e^{\lambda M_N}\1_{ \{\nu_N\in B\}}\bigr] 
			& \leq \E\bigl[ \e^{\lambda M_N}\1_{\{M_N/N\in A\}}\bigr] 
				+ \E\bigl[\e^{\lambda M_N} \1_{\{M_N/N\in A^{\rm c},\, \nu_N\in B\}} \bigr] \\
			& \leq \e^{- N c q \lambda + o(N)} 
				+ \e^{- N\inf_{(y,\nu)\in A^{\rm c}\times B} [\widetilde J(y,\nu) - \lambda y] +o(N)} 
	\end{aligned} 
	\ee	
	with the help of Varadhan's lemma. On $A^{\rm c} \times B$, we have 
	 \[
	 	\widetilde J(y,\nu) - \lambda y \geq \mu^{-1} \bigl(1 - C\sqrt{\lambda}\bigr)\Bigl( H(\nu;\mathcal{L}(T)) -\lambda\Bigr)\geq \lambda  \mu^{-1} \bigl(1 - C\sqrt{\lambda}\bigr). 
	 \]
	 Here we have used Pinsker's inequality 
	 \[
	 	H(\nu;\mathcal{L}(T))\geq 2 ||\nu- \mathcal{L}(T)||_\mathrm{TV}^2
\]
	 and the definition of $B$. The lemma now follows from~\eqref{eq:tilted-empir-aux}. 
\end{proof} 

Lemmas~\ref{lem:tilted} and~\ref{lem:tilted-empirical} are formulated in terms of the variable $T$ only. Combined with the information that $T$ is close to the geometric variable $G$ introduced in Section~\ref{sec:geometric}, we obtain the following. 

\begin{prop} \label{prop:tilted}
	There exists $\delta>0$ and $C,c>0$ such that the following holds for all $\eps,q,\lambda\in (0,\delta)$: 
	\begin{align*} 
		\limsup_{N\to \infty} \frac{1}{N}\log \E\Bigl[\e^{\lambda M_N} \1_{\{ |M_N/N -q| \geq C q \max(q,\eps,\sqrt{\lambda}) \} }\bigr]  & \leq - c q \max(\eps^2,\lambda), \\
		\limsup_{N\to \infty} \frac{1}{N}\log \E\Bigl[\e^{\lambda M_N} \1_{\{ ||\nu_N- \mathcal{L}(G)||_\mathrm{TV} \geq C\max(\eps,\sqrt{\lambda}) \} }\bigr]  & \leq - c q \max(\eps^2,\lambda). 
	\end{align*} 
\end{prop} 

\begin{proof}
	Let $C,c,\eps,q$ be as in Lemmas~\ref{lem:tilted} and~\ref{lem:tilted-empirical}. Let $C'>0$.
	As $\mu = \E[T]$ satisfies $\mu q = 1+q+O(\eps)$ by Lemma~\ref{lem:geom1}, on the event $|M_N/N - q|\geq C' q \max(q,\eps,\sqrt{\lambda})$ we have for sufficiently small $\eps$
	\[
		\Bigl|\mu\frac{M_N}{N} - 1\Bigr| \geq \mu\Bigl|\frac{M_N}{N} - q\Bigr| -|\mu q - 1| 
\geq \frac{1}{2} C' \max (q,\eps,\sqrt{\lambda})  
	\]
	if $C'$ is chosen sufficiently large. Similarly, on the event $||\nu_N- \mathcal{L}(G)||_\mathrm{TV}\geq C'\max(\eps,\sqrt{\lambda})$, by Lemma~\ref{lem:geom1}, we have for sufficiently small $\eps$ and large $C'$ 
	\[
		||\nu_N- \mathcal{L}(T)||_\mathrm{TV} \geq  C'\max(\eps,\sqrt{\lambda}) + ||\mathcal{L}(G) - \mathcal{L}(T)||_\mathrm{TV}\geq 
			\frac{1}{2} C' \max (\eps,\sqrt{\lambda}).
	\]
	If $\eps\leq \sqrt{\lambda}$ we conclude with Lemmas~\ref{lem:tilted} and~\ref{lem:tilted-empirical}. For $\eps\geq \sqrt{\lambda}$ we repeat the proofs of Lemmas~\ref{lem:tilted} and~\ref{lem:tilted-empirical} but with modified definitions of the sets $A$ and $B$ (replace $\sqrt{\lambda}$ by $\eps$).
\end{proof}

\section{Constant-pressure ensemble} \label{sec:pressure}

In this section we formulate and prove the analogs of Theorems~\ref{thm:fe} and~\ref{thm:lowdens} in the constant-pressure ensemble. Our main results in this section are Theorems~\ref{thm:ge} and \ref{thm:gibbsfe}.
Motivated by the heuristics from Section~\ref{sec:heuristics} we focus on $\beta$-dependent pressures $p_\beta$ such that  
\be \label{eq:pchoice}
	\beta p _\beta= \e^{- \beta e_\mathrm{surf}/2+o(\beta)} \qquad \mbox{as }\beta\to \infty. 
\ee
Recall that we write $\Q_N^\ssup{\beta}$ for the Gibbs measure in the constant-pressure ensemble, i.e., the probability measure on $\R_+^{N-1}$ with probability density ${\boldsymbol z}=(z_1,\dots,z_{N-1})\mapsto Q_N(\beta,p_\beta)^{-1}\exp(- \beta [U_N({\boldsymbol z})+ p_\beta \sum_{j=1}^{N-1} z_j])$, and the partition function $Q_N(\beta,p_\beta)$ is given by \eqref{QNdef}.

\subsection{Effective model} \label{sec:effective-model}

As in \eqref{truncpartfct}, define the truncated partition function by
\bes
	Q_N^\cl(\beta,p) = \int_{[0,R]^{N-1}} \e^{-\beta [U_N(z_1,\ldots,z_{N-1}) + p \sum_{j=1}^{N-1}z_j]} \dd z_1\cdots \dd z_{N-1}
\ees
when $N\geq 2$, and set $Q_1^\cl(\beta,p):=1$. Recall that  Assumption~\ref{ass:v1}(i) and~\ref{ass:R} imply that $v(r)\leq 0$ for $r\geq R$. As interactions across cracks are always zero or negative, discarding them decreases the Boltzmann weight $\exp( - \beta U)$. Summing over the number and locations of cracks, we get the inequality 
\be \label{eq:rec1}
	Q_N(\beta,p) \geq  Q_N^\cl(\beta,p) + \sum_{n=1}^{N-1}\Bigl( \int_R^\infty \e^{- \beta p r} \dd r\Bigr)^n  \sum_{1\leq i_1<\cdots<i_n\leq N-1} \prod_{k=1}^{n+1} Q_{i_k- i_{k-1}}^\cl(\beta,p)
\ee 
with the conventions $i_0=0$ and $i_{n+1}= N$. 
(If $v=0$ on $(R,\infty)$, the inequality~\eqref{eq:rec1} is in fact an equality.)
With $g^\cl(\beta,p_\beta)$, $g^\cl_\mathrm{surf}(\beta,p_\beta)$ as in Proposition~\ref{prop-oldresults} we define 
\begin{eqnarray}
	V_\beta( k) &=& -\frac{1}{\beta} \log Q_k^\cl(\beta,p_\beta) - k g^\cl(\beta,p_\beta) - g^\cl_\mathrm{surf}(\beta,p_\beta) \label{eq:veff} \\
	f_{\beta}(k) & =& \exp( - \beta V_\beta(k))- 1 \label{eq:fdef}\\
	q_\beta  & =& \frac{\exp(-\beta [g_\mathrm{surf}^\cl (\beta,p_\beta)+p_\beta R]) }{\beta p_\beta}
	 \label{eq:qdef}.
\end{eqnarray}
For simplicity we suppress the $R$-dependence from the notation for $V_\beta$, $f_\beta$, and $q_\beta$.  In this notation, Eq.~\eqref{eq:rec1} becomes 
\begin{equation} \label{eq:rec2}
	Q_N(\beta,p_\beta)\frac{\exp(- \beta p_\beta R)}{\beta p_\beta} \geq  \exp\bigl(- \beta N g^\cl(\beta,p_\beta)\bigr) \sum_{n=0}^{N-1} q_\beta^{n+1}  \sum_{0=i_0< i_1<\cdots<i_{n+1}= N} \exp \Bigl( - \beta \sum_{k=1}^{n+1} V_\beta(i_{k} - i_{k-1})\Bigr). 
\end{equation} 
Let $T,T_1,\ldots$ be i.i.d.~random variables as in Section~\ref{sec:aux} with $f=f_{\beta}$ and $q=q_\beta$. The $\beta,R$-dependence is suppressed from the notation. Think of $T_k = i_k- i_{k-1}$ in Eq.~\eqref{eq:rec2}. 
Then 
\be\label{eq:rec3}
	 Q_N(\beta,p_\beta)\frac{\exp(-\beta p_\beta R)}{\beta p_\beta}\, \e^{\beta N g^\cl(\beta,p_\beta)} u_\beta^N \geq \sum_{n=0}^{N-1} \P(T_1+\cdots+ T_{n+1} = N),
\ee
with  $u_\beta = u$  as in Eq.~\eqref{eq:renewalsolv}. For an upper bound we use Assumption~\ref{ass:v1}(v); set 
\be \label{eq:ladef}
	\lambda_\beta := \log \frac{\int_R^\infty \exp(\beta [C r^{-(s-2)} - p_\beta r]) \dd r}{\int_R^\infty \exp( - \beta p_\beta r) \dd r} 
\ee
with $C$ as in Eq.~\eqref{eq:across} for $v$ with unbounded support and $\lambda_\beta = 0$ if $v$ has compact support. For $n\in \N_0$ and $0=i_0<\ldots<i_{n+1} = N$ let $B_N(i_1,\ldots,i_n)$ be the event that there are exactly $n$ cracks, located after the particles with labels $i_1,\ldots,i_n$, as in \eqref{Bndef}. As in \eqref{clusternumber} we denote by $M_N$ the number of clusters in a configuration. To avoid confusion we write $M_N^{\rm lg}$ for the lattice gas variable from Eq.~\eqref{eq:mnk}. Also recall the generating function $\varphi=\varphi_\beta$ from Eq.~\eqref{eq:phidef}. 

\begin{lemma} \label{lem:rep}
Under Assumptions~\ref{ass:v1}(i)--(v), for all $\beta>0$,
\be\label{eq:rep1}
	-\beta g^\cl(\beta,p_\beta) - \log u_\beta \leq - \beta g(\beta,p_\beta) \leq -\beta g^\cl(\beta,p_\beta) -  \log u_\beta - \varphi_\beta^{-1}(- \lambda_\beta). 
\ee	
Moreover for all $n\in \N_0$ and $0=i_0<\cdots < i_{n+1}=N$, 
\be \label{eq:rec5}
	\Q_N^\ssup{\beta, p_\beta}\bigl( B_N(i_1,\ldots,i_n)\bigr) 
		\leq \e^{\lambda_\beta n} \frac{\P(T_1=i_1-i_0,\ldots,T_{n+1}= i_{n+1} - i_n)}{\sum_{\ell=1}^N \P(T_1 + \cdots + T_\ell = N)}. 
\ee	
In particular,
\be\label{eq:rec6}
	\Q_N^\ssup{\beta, p_\beta}(M_N = n+1) \leq \e^{\lambda_\beta n} \P(M_N^\mathrm{lg} = n+1).
\ee
\end{lemma} 

\begin{proof} 
	Estimating interactions across cracks with Eq.~\eqref{eq:across} we have 
	\begin{equation}\label{summing}
	\begin{aligned} 
	 \int_{\R_+^{N-1}} &\1_{B_N(i_1,\ldots,i_n)} (\vect{z})\e^{-\beta [U_N(\vect{z}) + p_\beta \sum_{j=1}^{N-1} z_j]} \dd \vect{z} \\
		&\quad \leq \Bigl(\e^{\lambda_\beta} \int_R^\infty \e^{- \beta p_\beta r} \dd r\Bigr)^n \prod_{k=1}^{n+1} Q_{i_k- i_{k-1}}^\cl(\beta,p_\beta) \\
		&\quad = \e^{\lambda_\beta n} \Bigl( \frac{\exp(-\beta p_\beta R)}{\beta p_\beta} \e^{\beta N g^\cl(\beta,p_\beta)} u_\beta^{N} \Bigr)^{-1} \, \P\bigl(T_1=i_1-i_0,\ldots,T_{n+1}= i_{n+1} - i_n\bigr). 
	\end{aligned} 
	\end{equation}
	We divide by $Q_N(\beta,p_\beta)$, combine with the lower bound~\eqref{eq:rec3} and obtain the inequality~\eqref{eq:rec5}. The inequality~\eqref{eq:rec6} follows by summing~\eqref{eq:rec5} over all increasing sequences $1\leq i_1<\cdots <i_{n}\leq N$. 
	
On the other  hand, summing \eqref{summing} over all increasing subsequences $1\leq i_1<\cdots <i_{n}\leq N$ and over $n\in\N$
	we find 
	\[ %\label{eq:rec4}
	\begin{aligned}
		Q_N(\beta,p_\beta) & \leq Q_N^\cl(\beta,p_\beta) + \sum_{n=1}^{N-1} \Bigl(\e^{\lambda_\beta} \int_R^\infty \e^{- \beta p_\beta r} \dd r\Bigr)^n\sum_{1\leq i_1<\cdots<i_n\leq N-1}  \prod_{k=1}^{n+1} Q_{i_k- i_{k-1}}^\cl(\beta, p_\beta) \\
		& = \sum_{n=0}^{N-1} \e^{\lambda_\beta n} \Bigl( \frac{\exp(-\beta p_\beta R)}{\beta p_\beta} \e^{\beta N g^\cl(\beta,p_\beta)} u_\beta^{N} \Bigr)^{-1} \, \P(T_1+\cdots+ T_{n+1} =N)\\
		& = \e^{-\lambda_\beta} \Bigl( \frac{\exp(-\beta p_\beta R)}{\beta p_\beta} \e^{\beta N g^\cl(\beta,p_\beta)} u_\beta^{N} \Bigr)^{-1} \, \E\bigl[\e^{\lambda_\beta M_N^\mathrm{lg}}\bigr]\P(\exists n\colon \,T_1+\cdots+ T_n =N).
\end{aligned}
	\]
	Now~\eqref{eq:rep1} follows from this, Eq.~\eqref{eq:rec2}, Lemma~\ref{lem:mnldp}, and the standard renewal result $\lim_{N\to \infty}\P(\exists n\colon T_1+\cdots + T_n=N) = \mu^{-1}\in (0,\infty)$.  
\end{proof} 

For later purpose we formulate a similar bound for the empirical distribution $\widehat \nu_N$ of the crack lengths defined in~\eqref{empdist-card-length}. Let $Y_i$, $i\in \N$, be i.i.d.\ exponential random variables with parameter $\beta p_\beta$, defined without loss of generality on the same underlying probability space as the lattice gas variable $M_N^\mathrm{lg}$. The $Y_i$'s are assumed to be independent of $M_N^\mathrm{lg}$. Define 
\be\label{eq:muhatdef}
	\widehat \nu_N^\mathrm{lg}:= \frac{1}{M_N^\mathrm{lg}-1} \sum_{i=1}^{M_N^\mathrm{lg} - 1} \delta_{Y_i}.
\ee

\begin{lemma}\label{lem:rep2}
	Under Assumptions~\ref{ass:v1}(i)--(v), for all $\beta>0$, for all $n\in \N_0$, and $D$ a measurable subset of the set of probability measures on $\R_+$, we have 
	\[ 
		\Q_N^\ssup{\beta, p_\beta}\bigl(\widehat \nu_N\in D,\, M_N = n+1\bigr)
		\leq \E\Biggl[ \exp\Bigl( C\beta (M_N^\mathrm{lg} - 1) \int_0^\infty (R+y)^{-(s-2)} \dd \widehat \nu_N^\mathrm{lg}(y) \Bigr)\1_{\{\widehat \nu_N^\mathrm{lg} \in D\}} \1_{\{M_N^\mathrm{lg} = n+1\}}\Biggr]
	\] 
\end{lemma} 

\begin{proof} 
	We use the notation of Lemma~\ref{lem:rep} and its proof.  Refining the first inequality in~\eqref{summing}, we see
	\begin{multline}\label{summing2}
	 \int_{\R_+^{N-1}} \1_{B_N(i_1,\ldots,i_n)} (\vect{z})\e^{-\beta [U_N(\vect{z}) + p_\beta \sum_{j=1}^{N-1} z_j]} \1_{\{\widehat \nu_N\in D\}} \dd \vect{z} \\
		 \leq \Biggl( \int_{[0,\infty)^n} \exp\Bigl( \beta C \sum_{j=1}^n (R+r_j)^{- (s-2)} - \beta p_\beta \sum_{j=1}^n (r_j +R) \Bigr)\1_{\{ \frac1n \sum_{j=1}^n \delta_{r_j}\in D\}} \dd \boldsymbol{r} \Biggr)\\
		 	\times  \prod_{k=1}^{n+1} Q_{i_k- i_{k-1}}^\cl(\beta,p_\beta) 
	\end{multline}
	with the help of \eqref{eq:across}. The expression in the second line is rewritten with the aid of the i.i.d.\ exponential random variables with probability density function $\beta p_\beta \exp( - \beta p_\beta r)$ as 
	\[
		\frac{\exp(- n \beta p_\beta R)}{(\beta p_\beta)^n} \E\Bigl[ \exp\Bigl( C\beta \sum_{i=1}^n (R+Y_i)^{-(s-2)} \Bigr)\1_{\{\frac1 n\sum_{i=1}^n \delta_{Y_i}\in D\}}\Bigr].   
	\] 		
	Substituting this expression in the second line of~\eqref{summing2} we obtain 
	\[
	\begin{aligned} 
	 \int_{\R_+^{N-1}} &\1_{B_N(i_1,\ldots,i_n)} (\vect{z})\e^{-\beta [U_N(\vect{z}) + p_\beta \sum_{j=1}^{N-1} z_j]} \1_{\{\widehat \nu_N\in D\}} \dd \vect{z}\\
	 &\leq \E\Bigl[ \exp\Bigl( C\beta \sum_{i=1}^n (R+Y_i)^{-(s-2)} \Bigr)\1_{\{\frac1 n\sum_{i=1}^n \delta_{Y_i}\in D\}}\Bigr]\\
	 & \qquad \qquad \times
	 \Bigl( \frac{\exp(-\beta p_\beta R)}{\beta p_\beta} \e^{\beta N g^\cl(\beta,p_\beta)} u_\beta^{N} \Bigr)^{-1} \, \P\bigl(T_1=i_1-i_0,\ldots,T_{n+1}= i_{n+1} - i_n\bigr), 
	\end{aligned}
	\]
	compare the third line in~\eqref{summing}. 	We divide by $Q_N(\beta,p_\beta)$, combine with the lower bound~\eqref{eq:rec3} and obtain an inequality similar to~\eqref{eq:rec5}:
	\begin{multline}
		\Q_N^\ssup{\beta, p_\beta}\bigl( \{\widehat \nu_N\in D\}\cap  B_N(i_1,\ldots,i_n)\bigr) \\
		\leq \E\Bigl[ \exp\Bigl( C\beta \sum_{i=1}^n (R+Y_i)^{-(s-2)} \Bigr)\1_{\{\frac1 n\sum_{i=1}^n \delta_{Y_i}\in D\}}\Bigr]\, \frac{\P(T_1=i_1-i_0,\ldots,T_{n+1}= i_{n+1} - i_n)}{\sum_{\ell=1}^N \P(T_1 + \cdots + T_\ell = N)}. 
	\end{multline} 
	We sum over $i_1,\ldots, i_n$, remember the definition~\eqref{eq:mnk} of the distribution of the lattice gas variable $M_N^\mathrm{lg}$, and exploit the independence of $Y_i$ and $M_N^\mathrm{lg}$. This gives 
	\[
	\begin{aligned}
		&\Q_N^\ssup{\beta, p_\beta}\bigl(\widehat \nu_N\in D,\, M_N = n+1\bigr) \\
		&\qquad \leq \E\Bigl[ \exp\Bigl( C\beta \sum_{i=1}^n (R+Y_i)^{-(s-2)} \Bigr)\1_{\{\frac1 n\sum_{i=1}^n \delta_{Y_i}\in D\}}\Bigr]\, \P(M_N^\mathrm{lg} =n+1)\\
		&\qquad = \E\Biggl[ \exp\Bigl( C\beta \sum_{i=1}^n (R+Y_i)^{-(s-2)} \Bigr)\1_{\{\frac1 n\sum_{i=1}^n \delta_{Y_i}\in D\}} \1_{\{M_N^\mathrm{lg} = n+1\}}\Biggr].
	\end{aligned}
	\]
	To conclude, we note that on the event $M_N^\mathrm{lg} =n+1$, we have $\frac1n\sum_{i=1}^{n}\delta_{Y_i} = \widehat \nu_N^\mathrm{lg}$ and 
	\[ 
		\sum_{i=1}^n (R+Y_i)^{-(s-2)} = (M_N^\mathrm{lg} - 1) \int_0^\infty (R+y)^{-(s-2)} \dd \widehat \nu_N^\mathrm{lg}(y). \qedhere
	\] 
\end{proof} 

\subsection{Bounds on effective quantities} 

In order to apply the results from Section~\ref{sec:aux}, we need to check that $q_\beta$, $\lambda_\beta$, and $q_\beta \sum_{k=1}^\infty |f_\beta(k)|$ are small. We start with~$q_\beta$. 
 
\begin{lemma} \label{lem:qbetasy}
	Under Assumptions~\ref{ass:v1}(i)--(v), we have 
	$$
	 \lim_{\beta\to \infty} g^\cl(\beta,p_\beta) = e_0,\quad \lim_{\beta\to \infty} g^\cl_\mathrm{surf}(\beta,p_\beta) = 	e_\mathrm{surf}>0,\quad q_\beta = \e^{-\beta e_\mathrm{surf}/2 +o(\beta)}.
	$$
\end{lemma} 

\begin{proof}
In our previous work \cite[Theorem 2.5]{jkst19} we investigated the asymptotic behavior of $g(\beta,p)$ and $g_\mathrm{surf}(\beta,p)$ at fixed $p>0$ and for the full partition function $Q_N(\beta,p)$. The strict positivity of the pressure was needed to ensure exponential tightness of measures on $\R_+^{\N}$ or $\R_+^\Z$ as $\beta \to \infty$. For restricted partition functions with spacings in $[0,R]$, exponential tightness comes for free and the results extend to vanishing pressure $p_\beta \to 0$. The asymptotic relations for $g^\ssup{R}(\beta,p_\beta)$ and $g^\ssup{R}_\mathrm{surf}(\beta,p_\beta)$ follow. Together with the definition~\eqref{eq:qdef} of $q_\beta$ and our choice of pressure~\eqref{eq:pchoice}, this implies the asymptotic behavior of $q_\beta$. 	 
\end{proof} 

Next we estimate $\lambda_\beta$ defined in Eq.~\eqref{eq:ladef}. The following lemma crucially needs Assumption~\ref{ass:R} on the size of the truncation parameter $R$.

\begin{lemma} \label{lem:gapintegral}
	Under Assumptions~\ref{ass:v1}--\ref{ass:R}, there exists $c>0$ such that $\lambda_\beta = O(\e^{-c\beta})$. 
\end{lemma}

\begin{proof}
	Clearly $\lambda_\beta \geq 0$. 
 For an upper bound, we first observe that $\exp( \beta C r^{-(s-2)}) \leq \exp( \delta)$ if and only if $r \geq (C \beta / \delta)^{1/(s-2)}$ and split the integral accordingly:
  \bes
  \begin{aligned}
  &    \beta p_\beta \int_R^\infty \e^{\beta [C r^{-(s-2)} - p_\beta r]} \dd r \\
	&\qquad  \leq \beta p_\beta  \int_0^{( C\beta /\delta)^{1/(s-2)}} \e^{\beta C R^{-(s-2)}} \e^{-\beta p_\beta r} \dd r
	   +  \beta p_\beta \int_{(C \beta /\delta)^{1/(s-2)}}^\infty \e^{\delta}\e^{-\beta p_\beta r } \dd r \\
	&\qquad  \leq \e^{\beta C R^{-(s-2)}} \bigl(1- \e^{-\beta p_\beta [\beta C/\delta]^{1/(s-2)}}\bigr) + \e^\delta \\
	& \qquad = O\Bigl(\e^{\beta C R^{-(s-2)}} \beta p_\beta (\tfrac{\beta C}{\delta})^{1/(s-2)}\Bigr) + 1+ O(\delta).
  \end{aligned}
  \ees
 To conclude, we choose $\delta_\beta= \exp( - c_1 \beta)$ such that 
 $c_2:=\min (c_1, e_\mathrm{surf}/2 - C R^{-(s-2)}- (s-2)^{-1} c_1)$ is positive (which exists because of Assumption~\ref{ass:R}),
 take the logarithm, use the assumption~\eqref{eq:pchoice} on the pressure,  and obtain an upper bound $\lambda_\beta \leq \exp( - \beta c_2+ o(\beta))$. 
\end{proof}

Next we make sure that $\eps_\beta =q_\beta  \sum_{k=1}^\infty |f_\beta (k)|$ vanishes for $\beta\to\infty$. Let us first check a necessary condition. Notice that for every fixed $k$
\bes
	f_{\beta}(k) = \e^{- \beta [E_k - k e_0- e_\mathrm{surf} + o(1)]} - 1,\qquad\mbox{as }\beta\to\infty.
\ees
In view of Lemma~\ref{lem:qbetasy}, in order that at least $\lim_{\beta\to\infty}q_\beta\sup_k |f_\beta(k)|= 0$, it is necessary that $\frac{1}{2}e_\mathrm{surf}  - (E_k - k e_0)<0$ for all $k\in \N$. The following lemma implies this. It is related to bounds derived in~\cite{scardia-schloem-zanini11}, see also~\cite[Remark 2.3]{braides-cicalese07}. 

%Let us recall that $e_0$ and $e_{\rm surf}$ depend on the interaction range parameter $m$ that we put equal to 2 in Assumption~\ref{ass:m}. However, the following lemma also holds under the weaker assumption that $m=2$ or $m\geq 3$ and $v(a)<\sum_{k=3}^m (k-2) v(ka)$. 

\begin{lemma} \label{lem:surfbound}
Suppose that Assumptions~\ref{ass:v1}(i)--(v) and~\ref{ass:m} hold true. Then 
\begin{equation}\label{Eklowerbound}
  E_n - n e_0\geq |e_0| > e_\mathrm{surf}/2,\qquad n\in\N.
\end{equation}
\end{lemma} 

\begin{proof} 
For $n=1$ we have $E_1=0$, and the first inequality of \eqref{Eklowerbound} is trivial. By~\cite[Lemma~3.2]{jkst19}, the sequence $(E_{n+1}/n)_{n\in \N}$ is subadditive and $e_0 = \inf_{n\in \N} (E_{n+1}/n)$. Hence, we have, for every $n\in\N\setminus\{1\}$,
\bes
	E_n - n e_0 = E_{n-1+1} - (n-1) e_0 - e_0 \geq - e_0 = |e_0|.
\ees 
From~\cite[Theorem 2.2]{jkst19}, we know that the surface energy is smaller than the clamped surface energy, i.e., $e_\mathrm{surf}\leq - \sum_{k=1}^m k v(ka)$. Therefore 
\bes 
	e_\mathrm{surf} +2 e_0 \leq - \sum_{k=1}^m k v(ka) +2 \sum_{k=1}^m v(ka) \leq v(a) - \sum_{k=2}^m (k-2) v(ka).
\ees
Since $m=2$, the right-hand side equals $v(a)$, which is negative. 
\end{proof} 

\begin{remark} 
The proof shows that Lemma~\ref{lem:surfbound}, and thus also the following Lemma~\ref{lem:qsurfbound}, is valid for $m\in\N\cup\{\infty\}$ provided $v$ satisfies the estimate $v(a)<\sum_{k=3}^m (k-2) v(ka)$. This is e.g.\ true for the Lennard-Jones potential $v(r) = r^{- 2s} - r^{-s}$ provided $s\geq 3$. 
\end{remark} 

For the following lemma, we also need Assumptions~\ref{ass:v1}(vi) and~\ref{ass:R}. 

\begin{lemma} \label{lem:qsurfbound}
	Under Assumptions~\ref{ass:v1}--\ref{ass:m}, we have
	$$
	\limsup_{\beta\to \infty}\sup_{n\in \N} \left(\frac{1}{\beta} \log Q_{n}^\ssup{R}(\beta,p_\beta) + n g^\ssup{R}(\beta,p_\beta) \right) 
	\leq e_0.
	$$ 
\end{lemma} 

\begin{proof} 
	Let $k,n \in \N_0$. If $U_{k+n+1}(z_1,\ldots,z_{k+n}) <\infty$, then every spacing $z_j$ must be larger than $r_\mathrm{hc}$ and all interactions involving more than one bond $z_j$, e.g., $v(z_j+z_{j+1})$, are zero or negative because of Assumption~\ref{ass:v1}(vi). It follows that 
	$$
	Q_{k+n+1}^\ssup{R}(\beta,p_\beta) \geq Q_{k+1}^\ssup{R}(\beta,p_\beta) Q_{n+1}^\ssup{R}(\beta,p_\beta),\qquad k,n \in \N.
	$$
	Hence 
	$ \exp(-\beta g^\ssup{R}(\beta,p_\beta)) \geq Q_{n+1}^\ssup{R}(\beta,p_\beta)^{1/n}$ 
	 for all $n\in \N$. Consequently
	 \bes
	  \sup_{n\in\N} \left(\frac{1}{\beta} \log Q_{n}^\ssup{R}(\beta,p_\beta) + n g^\ssup{R}(\beta,p_\beta) \right) 
	 		\leq g^\ssup{R}(\beta,p_\beta) = e_0+o(1).  \qedhere
	\ees
\end{proof} 

\begin{theorem} \label{thm:codec}
	Under Assumptions~\ref{ass:v1}--\ref{ass:m}, we have
	\begin{equation}\label{fksumupperbound}
	\limsup_{\beta\to \infty}\frac{1}{\beta}\log\Bigl( \sum_{k=1}^\infty |f_\beta(k)|\Bigr) \leq e_0+ e_\mathrm{surf}.
	\end{equation}
\end{theorem} 

Since $|e_0|> e_\mathrm{surf}/2$ by Lemma~\ref{lem:surfbound} and $q_\beta = \exp( - \beta[e_\mathrm{surf}/2+o(1)])$ by Lemma~\ref{lem:qbetasy}, we obtain right away the following corollary. 

\begin{cor} \label{cor:esmall} 
	Under Assumptions~\ref{ass:v1}--\ref{ass:m},
	$$
		\eps_\beta:= q_\beta \sum_{k=1}^\infty |f_\beta(k)| = O(\e^{-\beta (|e_0|-e_\mathrm{surf}/2)})\to 0. 
	$$ 
\end{cor}

In particular, $q_\beta$ and $f_\beta$ satisfy the condition~\eqref{ass:f}, and the results from Section~\ref{sec:aux} are applicable. 
%\noindent  In particular, the statements (a)-(c) of Lemma~\ref{lem:ibound} \comment{modify} hold true. 
%Since these statements were the only form in which the condition~\ref{ass:f} entered, we may apply the general theorems on the weakly interacting lattice gas from Section~\ref{sec:aux}. 
As a preparation for the following proofs, we note that the shift-invariant restricted Gibbs measure $\mu^\ssup{R}_\beta$ at pressure $p$ on $[0,R]^\Z$ is given by 
\bes
	\lim_{k\to \infty} \int_{[0,R]^{\ell}} f(z_{i_k+1},\ldots,z_{i_k+\ell}) \dd \Q_k^\ssup{R,\beta,p}(z_1,\ldots,z_{k-1})  = \int_{[0,R]^{\Z}} f(z_1,\ldots,z_\ell) \dd \mu^\ssup{R}_\beta( (z_j)_{j\in \Z})
\ees
for all $f \in C([0,R]^\ell)$ and sequences $i_k$ with $i_k\to \infty$ and $k-i_k\to \infty$. (The reader is referred to \cite{jkst19} for details on $\mu_\beta^\ssup{R}$.) 
In terms of the interaction 
\bes 
    \mathcal{W}_n((z_j)_{j\in\Z}) 
    = \sum_{\heap{j \leq n, k \geq n+1}{|k-j|\leq m-1}} v(z_j+\cdots + z_k)
\ees 
between a left and a right part of an infinite chain (in particular, $\mathcal{W}_n((z_j)_{j\in\Z}) = v(z_n + z_{n+1})$ if $m = 2$) in \cite[Proposition~4.9]{jkst19} and its proof one finds the explicit formulae   
\bes
	Q_{n+1}^\ssup{R}(\beta,p) 
	= \e^{-\beta n g^\ssup{R}(\beta,p)} \times \frac{\mu^\ssup{R}_\beta( \e^{ \beta [\mathcal{W}_0+ \mathcal{W}_n] })}{\mu^\ssup{R}_\beta(\e^{ \beta \mathcal{W}_0 })}, \qquad 
	\e^{-\beta g_\mathrm{surf}^\ssup{R}(\beta,p_\beta)}  
	=\e^{\beta g^\ssup{R}(\beta,p)} \mu^\ssup{R}_\beta(\beta \mathcal{W}_0)  
\ees
and
\be\label{eq:Q-mu-R}
    \Q_{n+1}^\ssup{R,\beta,p}(f) = \frac{\mu^\ssup{R}_\beta( f \e^{ \beta [\mathcal{W}_0+ \mathcal{W}_n] })}{\mu^\ssup{R}_\beta(\e^{ \beta [\mathcal{W}_0+ \mathcal{W}_n]} )} 
\ee
for $f \in C([0,R]^\ell)$ whenever $n \ge m-1$ for the unrestricted quantities $Q_{n+1}^\ssup{\beta,p}, \Q_{n+1}(\beta,p)$ which directly transfer to $Q_{n+1}^\ssup{R}(\beta,p), \Q_{n+1}^\ssup{R,\beta,p}$. 

\begin{proof} [Proof of Theorem~\ref{thm:codec}] Because of our restriction to spacings in $[0,R]$, the results from~\cite{jkst19} extend to vanishing pressure $p_\beta \to 0$. This holds true in particular for~\cite[Theorem 2.11]{jkst19}, which together with Proposition~4.9 and its proof in \cite{jkst19} shows the existence of some constants $c,\gamma >0$ such that
	$$ 
	|f_\beta(k)| \leq \e^{c\beta} \e^{- \gamma k},\qquad k\in\N.
	$$
To see this, we note that \cite[Theorem 2.11]{jkst19} gives 
\bes
|\mu^\ssup{R}_\beta( \e^{ \beta [\mathcal{W}_0+ \mathcal{W}_n] }) - \mu^\ssup{R}_\beta( \e^{ \beta \mathcal{W}_0}) \mu^\ssup{R}_\beta( \e^{ \beta \mathcal{W}_k})| 
\le \e^{- \gamma k} || \e^{ \beta \mathcal{W}_0} ||_{\infty}^2  \leq \e^{c\beta} \e^{- \gamma k}.
\ees
The claim then follows with a possibly larger $c$ from 
\bes
|f_\beta(k+1)| 
= \Bigl| \frac{Q_{k+1}^\ssup{R}(\beta,p_\beta)}{\exp( - \beta [(k+1) g^\ssup{R}(\beta,p_\beta) + g_\mathrm{surf}(\beta,p_\beta)])} -1\Bigr| 
= \Bigl| \frac{\mu^\ssup{R}_\beta(\e^{( \beta [\mathcal{W}_0+ \mathcal{W}_n]})}{[\mu^\ssup{R}_\beta(\e^{\beta \mathcal{W}_0})]^2} -1 \Bigr|
\ees
and the shift-invariance of $\mu^\ssup{R}_\beta$. By Lemma~\ref{lem:qsurfbound}, we have 
	$$ 
	\sup_{k\in\N}|f_\beta(k)| =\sup_{k\in\N} \Bigl| \frac{Q_k^\ssup{R}(\beta,p_\beta)}{\exp( - \beta [k g^\ssup{R}(\beta,p_\beta) + g_\mathrm{surf}(\beta,p_\beta)])} -1\Bigr| \leq 1 + \e^{\beta [e_0 +e _\mathrm{surf}+ o(1)]} .
	$$
It follows that, for all $n\in \N$, by splitting the sum after the $n$-th summand, 
	$$ 
		\sum_{k=1}^\infty |f_\beta(k)| \leq  n (1+ \e^{\beta [e_0+e_\mathrm{surf} + o(1)]}) + 	\frac{\e^{c\beta - \gamma n}}{1- \exp(- \gamma)}. 
	$$
Choosing $n = C\beta$ for some sufficiently large constant $C>0$, \eqref{fksumupperbound} follows. 
\end{proof}

\subsection{Number of cracks and empirical distributions}

We again use the letter $M_N$ for the random variable that counts the number of clusters (number of intervals between cracks) in a finite chain, i.e., $M_N\colon\R_+^{N-1}\to \N_0$, defined by $M_N(z_1,\ldots,z_{N-1}) = \#\{j\colon z_j \ge R\} + 1$ as in Eq.~\eqref{clusternumber}. The corresponding empircal measures  $\nu_N$ and $\widehat \nu_N$ are defined in Eq.~\eqref{empdist-card-length}. Let $G_\beta$ be a geometric variable with law $\P(G_\beta = k) = q_\beta (1+q_\beta)^{-k}$, $k\in \N$. 

\begin{theorem} \label{thm:ge}
Suppose Assumptions~\ref{ass:v1}--\ref{ass:m} hold true. 
Let $p_\beta$ be as in Eq.~\eqref{eq:pchoice} and $q_\beta,\lambda_\beta,\eps_\beta$ as in~\eqref{eq:qdef}, \eqref{eq:ladef} and Corollary~\ref{cor:esmall}. Set $\delta_\beta:=\max(\sqrt{\lambda_\beta}, \eps_\beta)$. 
Then for suitable $c,C,\beta_0\geq 0$ and all $\beta\geq \beta_0$,  
	\begin{eqnarray}
		\limsup_{N\to \infty} \frac{1}{N}\log \Q_N^\ssup{\beta, p_\beta}\Bigl( \Bigl|\frac{M_N}{N} - q_\beta \Bigr|\geq  Cq_\beta \max\{q_\beta,\delta_\beta\} \Bigr)
			& \leq& - c q_\beta \delta_\beta^2, \label{eq:thm-ge-M}\\
		\limsup_{N\to \infty} \frac{1}{N}\log \Q_N^\ssup{\beta, p_\beta}\Bigl( || \nu_N - \mathrm{Geom}(\tfrac{q_\beta}{1+q_\beta})||_\mathrm{TV} \geq  C \delta_\beta \Bigr) &\leq&  - c q_\beta  \delta_\beta^2,\label{eq:thm-ge-nu}\\
		\limsup_{N\to \infty} \frac{1}{N}\log \Q_N^\ssup{\beta, p_\beta}\Bigl( ||\widehat \nu_N  - \mathrm{Exp}(\beta p_\beta) ||_\mathrm{TV} \geq  C \delta_\beta \Bigr)
		&  \leq& - c q_\beta \delta_\beta^2. \label{eq:thm-ge-nuhat}
	\end{eqnarray}
\end{theorem} 

It follows in particular that as $N\to\infty$, the probability of the event $M_N/N = q_\beta (1+ O(\delta_\beta ))$ converges to $1$. 

\begin{remark}
	The estimate on $\widehat \nu_N$ actually holds true for every $\delta_\beta \geq \sqrt{\lambda_\beta}$. Moreover, for compactly supported potentials $v$, we have $\lambda_\beta =0$ and for each $j$, the distribution of $z_j - R$ conditional on $z_j\geq R$ is exactly an exponential law with parameter $\beta p_\beta$.
\end{remark} 

\begin{proof} 
	To avoid confusion we write $M_N^{\rm lg}$ and $\nu_N^{\rm lg}$ for the auxiliary lattice gas variables defined in Eqs.~\eqref{eq:mnk} and~\eqref{eq:empdist-T}. The statements for $M_N/N$ and $\nu_N$ are consequences of Lemma~\ref{lem:rep} and Proposition~\ref{prop:tilted}. More precisely, with $A = \{n \in \N \colon |n - q_\beta N| \geq Cq_\beta N \max\{q_\beta,\delta_\beta\}\}$, one has 
	\bes 
	\Q_N^\ssup{\beta, p_\beta} ( M_N \in A ) 
	\leq \sum_{n \in A-1} \e^{\lambda_\beta n} \, \P(M_N^{\rm lg} = n) 
	= \e^{-\lambda_{\beta}} \E\bigl[\e^{\lambda_{\beta} M_N^{\rm lg}} \1_{\{M_N^{\rm lg} \in A\}}\bigr]
	\ees 
	by Lemma~\ref{lem:rep} so that \eqref{eq:thm-ge-M} follows from Proposition~\ref{prop:tilted}. Likewise, if $A$ denotes the set of probability measures $\pi$ on $\N$ for which $|| \pi - \mathrm{Geom}(\tfrac{q_\beta}{1+q_\beta})||_\mathrm{TV} \geq  C \delta_\beta$, then 
	\bes
	\begin{aligned}
	\Q_N^\ssup{\beta, p_\beta}(\nu_N \in A) 
	&= \sum_{n=0}^{N-1} \, \sum_{0 = i_1<\cdots<i_{n+1}= N} \1_{\{\frac{1}{n+1} \sum_{k=1}^{n+1} \delta_{i_k-i_{k-1}} \in A\}} \Q_N^\ssup{\beta, p_\beta}\bigl( B_N(i_1,\ldots,i_n)\bigr) \\ 
	&= \E\bigl[\e^{\lambda_{\beta} M_N^{\rm lg}} \1_{\{\nu_N^{\rm lg} \in A\}}\bigr]
	\end{aligned}
	\ees
	by Lemma~\ref{lem:rep} and \eqref{eq:thm-ge-nu} follows again from Proposition~\ref{prop:tilted}. 
	
	For the empirical distribution  of crack lengths, let  $Y_1,Y_2,\ldots$ be i.i.d. random variables with exponential law $Y_i \sim \mathrm{Exp}(\beta p_\beta)$. The variables are taken independent of $T_1,T_2,\ldots$ and $M_N^\mathrm{lg}$. By Lemma~\ref{lem:rep2} with $\widehat \nu_N^\mathrm{lg}$ as in \eqref{eq:muhatdef}, 
\[
\begin{aligned}
	\Q_N^\ssup{\beta, p_\beta}\bigl(M_N \in B,\, \widehat \nu_N \in D\bigr) 
		& \leq \sum_{n+1 \in B} \E\bigl[ \exp \Bigl({ \beta M_N^\mathrm{lg}} \int_0^\infty C(R+r)^{-(s-2)} \dd \widehat \nu_N^\mathrm{lg}(r) \Bigr) 
		\1_{\{M_N^\mathrm{lg} =n+1,\, \widehat \nu_N^\mathrm{lg}\in D\}}\bigr] \\
		& \leq  \Bigl(\E\bigl[ \e^{ 2 \beta \lambda_\beta M_N^\mathrm{lg}}\1_{\{M_N^\mathrm{lg}\in B\}}\bigr] \P\bigl(M_N^\mathrm{lg}\in B,\, \widehat \nu_N^\mathrm{lg}\in D\bigr)\Bigr)^{1/2} \\
\end{aligned}
\]
for every subset $B$ of $\N$ and every measurable set $D$ of probability measures on $\R_+$. 
Now, similarly as in \eqref{eq:ratefcn-cluster}, $(M_N^\mathrm{lg}/N, \widehat \nu_N^\mathrm{lg})$ satisfies a large deviations principle with speed $N$ and rate function 
\[
	\mathcal{J} (y,\widehat \nu) = J(y) + y H\bigl( \widehat \nu; \mathrm{Exp}(\beta p_\beta) \bigr)
\]
with $J(y)$ defined in Lemma~\ref{lem:mnldp}. The proof is completed as in Proposition~\ref{prop:tilted} (and Lemmas~\ref{lem:tilted} and~\ref{lem:tilted-empirical}). 
\end{proof}

\subsection{Gibbs free energy and stress-strain relation}

Let 
\[
	\ell(\beta,p) =\frac{\partial g}{\partial p}(\beta,p),\qquad
	\ell^\cl(\beta,p) =\frac{\partial g^\cl}{\partial p}(\beta,p),\qquad L_k^\cl(\beta,p) = -\frac{1}{\beta}\frac{\partial}{\partial p}\log Q_k^\cl(\beta,p).
\]
Then $L_k^\cl(\beta,p) = \int_{[0,R]^{k-1}} (z_1 + \ldots + z_{k-1}) \dd\Q^\ssup{R,\beta,p}_k$ is the expected length of a $k$-cluster at inverse temperature $\beta$ and pressure $p$ while $\ell(\beta,p)$ and $\ell^\cl(\beta,p)$ represent the average spacings between consecutive particles in a chain or cluster with infinitely many particles. (In Lemma~\ref{lem:cluster-lengths} below we will see that $\frac{1}{k} L_k^\cl(\beta,p) \to \ell^\cl(\beta,p)$ uniformly in $p$ as $k \to \infty$.) Recall that $q_\beta$ and $p_\beta$ are both of order $\exp(-\beta e_\mathrm{surf}/2+o(\beta))$. 

\begin{theorem} \label{thm:gibbsfe}
	Suppose Assumptions~\ref{ass:v1}--\ref{ass:m} hold true. 
	Then 
	\begin{equation*}
		g(\beta,p_\beta)   = g^\cl(\beta,p_\beta) - \frac{q_\beta}{\beta}  \bigl(1+ o(1)\bigr)= e_0^\ssup{R}(\beta) + a p_\beta - \frac{q_\beta}{\beta} +o(p_\beta) + o(q_\beta/\beta)
	\end{equation*}
	and 
	\begin{align*} 
		\ell(\beta,p_\beta)  
		& = \ell^\cl(\beta,p_\beta) + o(1)   
            + \frac{q_\beta}{\beta p_\beta} (1+o(1))
		  	 = a +o(1) + \frac{\exp( - \beta e_\mathrm{surf}^\ssup{R}(\beta))}{(\beta p_\beta)^2} (1+o(1)). 
	\end{align*} 
\end{theorem} 

\noindent The proof requires several lemmas. 

\begin{lemma} \label{lem:cluster-lengths} 
	Suppose Assumptions~\ref{ass:v1}(i)--(v) and~\ref{ass:m} hold true. 
	Assume that  $p_\beta \to 0$ as $\beta\to \infty$.  Then
	\begin{enumerate} 
		\item [(a)] $\lim_{\beta\to \infty} \sup_{p\in [0,p_\beta]} |\ell^\ssup{R}(\beta,p_\beta) - a| = 0$. 
		\item [(b)] For some $\beta_0,c>0$ and all $\beta\geq \beta_0$, 
	\be\label{eq:Lk-kl-est}
	 \sup_{p\in [0,p_\beta]} \sup_{k\in \N} |L_k^\ssup{R}(\beta,p) - k \ell^\ssup{R}(\beta,p)| \leq c \beta.
	\ee	
	\end{enumerate} 
\end{lemma} 

\begin{proof} 
	We first prove (b). Choose $\widetilde p_\beta \in [0,p_\beta]$ and let  $\mu_\beta^\ssup{R}$ be the restricted Gibbs measure on $[0,R]^\Z$ at pressure $\widetilde p_\beta$. By \eqref{eq:Q-mu-R} we have 
	\be\label{eq:LkR}
	L_{k+1}^\ssup{R} = \frac{\sum_{j=1}^{k} \mu_\beta^\ssup{R} (\e^{\beta \mathcal{W}_0}  z_j \e^{\beta \mathcal{W}_k})}{\mu_\beta^\ssup{R}(\e^{\beta[ \mathcal{W}_0+ \mathcal{W}_k]})}. 
	\ee
	with $\mathcal{W}_n((z_j)_{j\in\Z}) = v(z_n + z_{n+1})$.  Proceeding as in Theorem~\ref{thm:codec}, we get that for some $c>0$ and all sufficiently large $\beta$, 
	$$
	\Bigl|\mu_\beta^\ssup{R} \bigl(\e^{\beta \mathcal{W}_0} z_j \e^{\beta \mathcal{W}_k}\bigr) - \mu_\beta^\ssup{R} \bigl(\e^{\beta \mathcal{W}_0}\bigr) \mu_\beta^\ssup{R}\bigl( z_j \bigr) \mu_\beta^\ssup{R}\bigl( \e^{\beta \mathcal{W}_k}\bigr)\Bigr|
			\leq \e^{c\beta} \e^{-\gamma \min(j, k-j)}
	$$
	and  
	$$ \Bigl|\mu_\beta^\ssup{R} \bigl(\e^{\beta \mathcal{W}_0}  \e^{\beta \mathcal{W}_k}\bigr) - \mu_\beta^\ssup{R} \bigl(\e^{\beta \mathcal{W}_0}\bigr) \mu_\beta^\ssup{R}\bigl( \e^{\beta \mathcal{W}_k}\bigr) \Bigr|
			\leq \e^{c\beta} \e^{-\gamma k}. 
	$$
	Since $\mu_\beta^\ssup{R}$ is supported on $[0,R]^\Z$ and shift-invariant, for any $\ell$ we can estimate 
	\bes
	\begin{aligned}
	|L_{k+1}^\ssup{R}(\beta,p) - k \mu^\ssup{R}_\beta(z_0)| 
	&= \frac{\sum_{j=1}^k \mu_\beta^\ssup{R} \bigl(\e^{\beta \mathcal{W}_0} [ z_j - \mu^\ssup{R}_\beta(z_0) ] \e^{\beta \mathcal{W}_k}\bigr)}{\mu_\beta^\ssup{R}\bigl(\e^{\beta[ \mathcal{W}_0+ \mathcal{W}_k]}\bigr)} \\ 
	&\leq 4 R \ell + \sum_{j=\ell+1}^{k-\ell} \frac{\e^{c\beta} \e^{-\gamma \min(j, k-j)} +\mu^\ssup{R}_\beta(z_0)\e^{c\beta} \e^{-\gamma k}}{\bigl(\mu_\beta^\ssup{R}(\e^{\beta \mathcal{W}_0})\bigr)^2} \\ 
	&\leq 4 R \ell + \e^{(c + 2 ||\mathcal{W}_0||_\infty)\beta} (2 + R)\frac{\e^{-\gamma\ell}}{1 - \e^{-\gamma}}.\end{aligned}
	\ees
	With $\ell = \lceil (c + 2 ||\mathcal{W}_0||_\infty)\beta/\gamma \rceil$ we obtain $|L_{k+1}^\ssup{R}(\beta,p) - k \mu^\ssup{R}_\beta(z_0)| \leq C \beta$. The estimates are uniform in $p=\tilde p_\beta \in [0,p_\beta]$ because the constant $\gamma$ is and because $||\mathcal{W}_0||_\infty< \infty$. In particular we have $\frac{1}{k} L_{k}^\ssup{R}(\beta,p) \to \mu^\ssup{R}_\beta(z_0)$ uniformly in $p$ and, in combination with \eqref{e0Rdef}, $\mu^\ssup{R}_\beta(z_0) = \ell^\ssup{R}(\beta,p)$. Thus also \eqref{eq:Lk-kl-est} follows. 
	
	Part (a) is now a consequence of~\cite[Corollary 2.6]{jkst19} since $\ell^\ssup{R}(\beta,p) = \mu^\ssup{R}_\beta(z_0)$.  Because of the restriction to spacings $z_j\in [0,R]$, the corollary applies to $p=p_\beta\to 0$ as well. 

\end{proof} 

\begin{lemma} \label{lem:qsurfcomp}
	Under Assumption~\ref{ass:v1} and~\ref{ass:m}, we have  as $\beta\to\infty$, 
	$$ 
	g_\mathrm{surf}^\ssup{R}(\beta,p_\beta) = e_\mathrm{surf}^\ssup{R}(\beta) + O(\beta p_\beta)\qquad\mbox{and}\qquad q_\beta = (1+o(1))\frac{\exp(- \beta e_\mathrm{surf}^\ssup{R}(\beta))}{\beta p_\beta}.
	$$ 
\end{lemma} 

\begin{proof} 
	From the definition of $L_k^\ssup{R}(\beta,p)$ and $\ell^\ssup{R}(\beta,p)$ and Lemma~\ref{lem:cluster-lengths}(b), we get 
	\begin{multline*}
		\Bigl|\Bigl( - \frac{1}{\beta}\log Q_k^\ssup{R}(\beta,p_\beta) - k g^\ssup{R}(\beta,p)\Bigr) - \Bigl( - \frac{1}{\beta}\log Q_k^\ssup{R}(\beta,0) - k e_0^\ssup{R}(\beta)\Bigr)\Bigr| \\
	=\Bigl| \int_0^{p_\beta} \bigl( L_k^\ssup{R}(\beta,p)- k \ell^\ssup{R}(\beta,p)\bigr) \dd p \Bigr|\leq c \beta p_\beta
	\end{multline*} 
	for some $k$-independent $c$ and all sufficiently large $\beta$. Letting $k\to \infty$ we find $|g_\mathrm{surf}^\ssup{R}(\beta,p_\beta) -e_\mathrm{surf}(\beta)|\leq c \beta p_\beta$. This proves the first part of the lemma. The expression for $q_\beta$ follows from the definition~\eqref{eq:qdef} of $q_\beta$ and the fact that $\beta p_\beta = o(\beta^{-1})$ by our choice of $p_\beta$. 
\end{proof} 

\noindent In order to analyze the system length, we condition on the number $M_N$ of clusters and express the system length as a sum of (conditionally) independent random variables. Let $T_1,T_2,\dots$ be i.i.d.~random variables with law as in~\eqref{eq:iid}, representing the cluster cardinalities. Further let $X_i, Y_i$ be random variables with the following properties: 
 The $Y_i$'s are i.i.d. with law $Y_i \sim \mathrm{Exp}(\beta p_\beta)$. They are also independent of the $T_i$'s and the $X_i$'s. 
 The $X_i$'s are i.i.d. and satisfy 
	$$
		\P(X_i \in B \mid T_i =k) = \frac{1}{Q_k^\ssup{R}(\beta,p_\beta)} \int_{\R_+^{k-1}} \1_B(z_1+\cdots + z_{k-1})\e^{- \beta [U_k(\vect{z}) + p_\beta \sum_{j=1}^{k-1} z_j] } \dd \vect{z}
	$$
for all $k \in \N$ and measurable $B\subset \R_+$, and $\P(X_i =0 \mid T_i =1) = 1$. Notice $\E[X_i \mid T_i =k] = L_k^\ssup{R}(\beta,p_\beta)$. Let 
\be \label{eq:lengthrv}
	\Lambda_n = X_1 + (R+ Y_1) + X_2 +\cdots + (R+Y_{n-1}) + X_n. 
\ee
Then $\Lambda_n$ represents the system length conditional on the event $\{M_N=n\}$ that there are $n$ clusters, neglecting the effect of interactions across cracks. 

\begin{lemma} \label{lem:craleldp}
	Under the assumptions of Theorem~\ref{thm:gibbsfe},
	 there exists $c>0$ such that for all sufficiently large $\beta$,
	$$
\e^{\lambda_\beta n} \P\Bigl(\Bigl|\frac{1}{n}\sum_{i=1}^{n} Y_i-\frac{1}{\beta p_\beta}\Bigr| \geq \frac{c \sqrt{\lambda_\beta}}{\beta p_\beta} \Bigr) \leq 2 \e^{- n \lambda_\beta},\qquad n\in\N.
$$ 
\end{lemma}

\begin{proof}
	We use Markov's inequality: for $0 \leq t <\beta p_\beta/2$ and some $\beta$-independent constant $c>0$, we have 
	\begin{align*}
	\P\Bigl( \sum_{i=1}^n Y_i \geq \frac{n}{\beta p_\beta} (1 + \eps) \Bigr)
		 & \leq \e^{ - n t (1+\eps)/\beta p_\beta)}	\prod_{i=1}^n \E[\e^{t Y_i}]
		  = \e^{- n t (1+\eps) / (\beta p_\beta)} \Bigl(1- \frac{t}{\beta p_\beta}\Bigr)^{-n}\\
		  & \leq \exp\Bigl(  n \Bigl( - \eps \frac{t}{\beta p_\beta}+ \frac{1}{2} c  \bigl(\frac{t}{\beta p_\beta}\bigr)^2\Bigr)\Bigr).
	\end{align*}
	In the last line we have estimated $- \log (1-s) \leq s + \frac{c}{2} s^2$ for $|s|\leq 1/2$. Choosing $t = \beta p_\beta \eps / c$ (we may assume without loss of generality that $\eps/c \leq 1/2$), we obtain the upper bound $\exp(- n \eps^2 / (2c))$. A similar argument shows 
	$$ 
		\P\Bigl( \sum_{i=1}^n Y_i \leq \frac{n}{\beta p_\beta} (1 - \eps) \Bigr)
		\leq \e^{- n \eps^2 / (2c)}. 
	$$ 
	To conclude, we choose $\eps = \sqrt{4 \lambda c}$. 
\end{proof}

\begin{lemma} \label{lem:clulemean}
	Under the assumptions of Theorem~\ref{thm:gibbsfe}, $ \E[X_i] = (1+o(1)) \frac{a}{q_\beta}$ as $\beta \to \infty$. 
\end{lemma}

\begin{proof}
	We have 
	$$ \E[X_i] = \E [T] \sum_{k=1}^\infty  \P(\widetilde T=k) \frac{L_k^\ssup{R}(\beta,p_\beta)}{k},$$ 
	where $\widetilde T$ is the size-biased variable. Lemma~\ref{lem:geom1} tells us that $\E[T]\sim 1/q_\beta$. 
	 Since $L_k^\cl/k \leq R$ for all $k\in \N$, Lemma~\ref{lem:geom1} also shows 
	\bes
		\Bigl|\sum_{k=1}^\infty \bigl( \P(\widetilde T= k) - \P(\widetilde G=k) \bigr )\frac{L_k^\cl(\beta,p_\beta)}{k}\Bigr| \leq R || \mathcal{L}(\widetilde T)- \mathcal{L}(\widetilde G) ||_\mathrm{TV} = O(\eps_\beta) \to 0. 
	\ees
		By Lemma~\ref{lem:cluster-lengths}, for every fixed $k_1\in \N$, 
		\bes
			\Bigl|\sum_{k=1}^\infty \P(\widetilde T =k) \frac{L_k^\cl(\beta,p_\beta)}{k} - \ell^\ssup{R}(\beta,p_\beta)  \Bigr| 
				\leq c \beta\, \P (\widetilde G\leq k_1 \beta) + \frac{c}{k_1}. 
		\ees
		Since $q_\beta\to 0$ exponentially fast, we have
		\bes
			\beta\, \P(\widetilde G\leq k_1\beta) = \beta \sum_{k=1}^{\lfloor k_1\beta\rfloor} \frac{kq_\beta^2}{(1+q_\beta)^{k+1}} = O(\beta^2 q_\beta) \to 0 \quad (\beta\to \infty).
		\ees	
		We let first $\beta\to \infty$, then $k_1\to \infty$, and find 
		altogether $\E[X_i] \sim \frac{1}{q_\beta} \ell^\ssup{R}(\beta,p_\beta)$, and we conclude with Lemma~\ref{lem:cluster-lengths}(a).
\end{proof}

\begin{lemma}\label{lem:cluleldp}
	Under the assumptions of Theorem~\ref{thm:gibbsfe}, there exists $c>0$ such that for all sufficiently large $\beta$,
		$$ 
		\e^{\lambda_\beta n} \P\Bigl(\Bigl|\frac{1}{n}\sum_{i=1}^{n} X_i- \E[X_1]\Bigr| \geq \frac{c \sqrt{\lambda_\beta}}{q_\beta} \Bigr) \leq 2 \e^{- n \lambda_\beta},\qquad n\in\N.
		$$ 	
\end{lemma}

\begin{proof}
	We use Markov's exponential inequality as in the proof of Lemma~\ref{lem:craleldp}. We have,  for $|t|\ll q$, 
	\bes
		\bigl|\E[\e^{t X_i}] -1 - t\,\E[X_i] \bigr| \leq \frac{1}{2} t^2 \E\bigl[X_i^2 \e^{|t| X_i}\bigr] \leq \frac{1}{2} (t R)^2 \E[ T^2 \e^{|t| R T}]
	\ees
	Using~\eqref{eq:mocubo} with $c= 2 \sqrt c_\tau$ we find that 
	$$ 
	\log \E[\e^{tX_i}] \leq t \E[X_i]  + \frac{c^2 t^2 }{8 q_\beta^2}
	$$ 
	uniformly for small $q_\beta$ and $|t|\leq \tau q_\beta$ where $\tau \in (0,1)$ is fixed. 
	It follows that 
	\begin{align*} 
		\E[\e^{\lambda_\beta n}\1_{\{X_1+\cdots + X_n\geq n E[X_1] + n \frac{\eps}{q_\beta} \}}] 
		& \leq \e^{\lambda_\beta n } \e^{ - t n (\E[X_1] + \frac{\eps}{q_\beta })}\bigl(\E[\e^{t X_1}] \bigr)^n \\
		&\leq  \exp\Bigl( n \Bigl( \lambda_\beta +  c^2 \frac{t^2}{8 q_\beta^2}  - t \frac{\eps}{q_\beta}\Bigr)\Bigr). 
	\end{align*} 
	With $\eps = c \sqrt{\lambda_\beta}$ and $t = 4 q_{\beta}\sqrt{\lambda_\beta}/c$ the remaining part of the proof is analogous to Lemma~\ref{lem:craleldp} and is left to the reader. 
\end{proof}

\begin{proof}[Proof of Theorem~\ref{thm:gibbsfe}]
	By Lemma~\ref{lem:rep}, Eq.~\eqref{eq:gel} and Lemma~\ref{lem:tilted}, we have $g(\beta,p_\beta) = g^\cl(\beta,p_\beta) + \frac{1}{\beta}\log u_\beta + O(q_\beta\lambda_\beta /\beta)$. 
	By Lemma~\ref{lem:geom1} we have  $\log u_\beta =  - \log (1+ q_\beta + O(q_\beta \eps_\beta))$. The first identity in the asymptotic approximation of $g(\beta,p_\beta)$ follows. For the second identity, we note that 
	\bes
		g^\ssup{R}(\beta,p_\beta) - e_0^\ssup{R}(\beta) = \int_0^{p_\beta} \ell ^\ssup{R}(\beta,p)\dd p = (1+o(1)) a p_\beta
	\ees
	where we have used Lemma~\ref{lem:cluster-lengths}(a) and 
	 $e^\ssup{R}_0(\beta) = g^\cl(\beta,0)$. 
	 
	For the average spacing, we first note that a 
    reasoning analogous to Lemma~\ref{lem:rep} shows that for every $B\subset \R_+$, 
	\bes 
		\Q_N^\ssup{\beta, p_\beta}\Bigl( \Big\{\vect{z}\in \R_+^{N-1}\colon \sum_{j=1}^{N-1} z_j \in B\Big\}\,\Big|\, M_N = n + 1 \Bigr) \leq \e^{\lambda_\beta n} \P(\Lambda_{n+1} \in B)
	\ees
	with $\Lambda_n$ the random variable~\eqref{eq:lengthrv} and on the left-hand side $M_n$ stands for the number of clusters of $\vect{z}$. 	
	 Similarly to the proofs of Lemmas~\ref{lem:rep} and~\ref{lem:rep2}, we have
	\begin{align*}
		\Q_N^\ssup{\beta, p_\beta}\Bigl( \Big\{\vect{z}\in \R_+^{N-1}\colon \sum_{j=1}^{N-1} z_j\1_{\{z_j\leq R\}} \in B\Big\}\,\Big|\, M_N =  n +1 \Bigr)  & \leq \e^{\lambda_\beta n} \P\Bigl( \sum_{j=1}^{n+1}X_j \in B\Bigr),\\	
		\Q_N^\ssup{\beta, p_\beta}\Bigl( \Big\{\vect{z}\in \R_+^{N-1}\colon \sum_{j=1}^{N-1} z_j \1_{\{z_j>R\}} \in B\Big\}\,\Big|\, M_N =  n +1 \Bigr) &\leq \e^{\lambda_\beta n} \P\Bigl(\sum_{j=1}^{n} (R + Y_j) \in B\Bigr).	
	\end{align*}	
	In combination with Lemmas~\ref{lem:craleldp} and~\ref{lem:cluleldp}, this gives 
	\begin{align*}
		\Q_N^\ssup{\beta, p_\beta}\Bigl( \Big\{\vect{z}\in \R_+^{N-1}\colon \Bigl| \frac{1}{n+1}\sum_{j=1}^{N-1}  z_j\1_{\{z_j\leq R\}} - \E[X_1]\Bigr| \geq \frac{c\sqrt{\lambda_\beta}}{q_\beta} \Big\}\,\Big|\, M_N =  n +1 \Bigr)  & \leq 2 \e^{- (n+1) \lambda_\beta}, \\	
		\Q_N^\ssup{\beta, p_\beta}\Bigl( \Big\{\vect{z}\in \R_+^{N-1}\colon \Bigl|\frac 1 n\sum_{j=1}^{N-1} z_j \1_{\{z_j>R\}} - R - \frac1{\beta p_\beta}\Bigr|\geq \frac{c\sqrt{\lambda_\beta}}{\beta p_\beta}\,\Big|\, M_N =  n +1 \Bigr) &\leq 2\e^{-\lambda_\beta n}.	
	\end{align*}
	As a consequence, using the general inequality $\P(A^c\cap B^c) \geq 1 - \P(A) - \P(B)$, we get
	\begin{align*}	
	\Q_N^\ssup{\beta, p_\beta}\Bigl( \Big\{\vect{z}\in \R_+^{N-1}\colon& \Bigl| \sum_{j=1}^{N-1} z_j - (n+1) \E[X_1] - n \Bigl( R + \frac{1}{\beta p_\beta} \Bigr) \Bigr| \\ 
	&\leq (n+1) \frac{c \sqrt{\lambda_\beta}}{q_\beta} + n \frac{c\sqrt{\lambda_\beta}}{\beta p_\beta}\,\Big|\, M_N =  n +1 \Bigr) %\\
        \geq 1- 4 \e^{- \lambda_\beta n}.
	\end{align*} 
	By Theorem~\ref{thm:ge}, there exist $C>0$, $\tilde\delta_\beta = \max\{\delta_\beta,q_\beta\}>0$ with $\tilde\delta_\beta\to 0$ as $\beta\to \infty$ such that 
	$$\lim_{N\to \infty}\Q_N^\ssup{\beta, p_\beta} \bigl( |M_N - N q_\beta|\leq CN \tilde\delta_\beta q_\beta \bigr) = 1,$$
	i.e., $M_N/N\sim q_\beta (1+O(\tilde\delta_\beta))$ with a probability converging to $1$. Therefore with a probability converging to $1$ (under $\Q_N^\ssup{\beta,p_\beta}$, without conditioning on $M_N$), 
	\begin{align*} 
		\Bigl|\frac 1 N \sum_{j=1}^{N-1} z_j - \ell_0(\beta,p_\beta)\Bigr| 	&\leq \Bigl|\frac{M_N}{N} - q_\beta\Bigr| \, \ell_0(\beta,p_\beta)
			+ \frac{M_N}{N} c \sqrt{\lambda_\beta}\Bigl( \frac{1}{q_\beta}+ \frac{1}{\beta p_\beta}\Bigr)  + \frac{1}{\beta p_\beta N}\\
		&\leq O(\tilde\delta_\beta) \ell_0(\beta,p_\beta) 
			+ c\sqrt{\lambda_\beta} (1+O(\tilde\delta_\beta)) \Bigl( 1+ \frac{q_\beta}{\beta p_\beta}\Bigr)  + \frac{1}{\beta p_\beta N},  
	\end{align*} 
	where we have set $\ell_0(\beta,p_\beta) = q_\beta ( \E[X_1] + R + \frac{1}{\beta p_\beta})$. 
	On the other hand, because of the uniqueness and ergodicity with respect to shifts of the infinite volume Gibbs measure~\cite{georgii-book,jkst19}, standard results ensure that $L_N/N\to \ell(\beta,p_\beta)$ almost surely. It follows that 
	\[
		\Bigl|\ell(\beta, p_\beta) -  \ell_0(\beta,p_\beta)\Bigr|
		\leq O(\tilde\delta_\beta) q_\beta\E[X_1] + \frac{q_\beta}{\beta p_\beta} \bigl( O(\tilde\delta_\beta) + O(\sqrt{\lambda_\beta})\bigr)  + O(\sqrt{\lambda_\beta})
	\] 
	and thus 
	\[
		\ell(\beta,p_\beta) = (1+o(1)) q_\beta \E[X_1] + (1+o(1)) \frac{q_\beta}{\beta p_\beta}  + o(1).
	\] 
	To conclude, we use Lemma~\ref{lem:clulemean} for $\E[X_1]$ and Lemma~\ref{lem:qsurfcomp} for $q_\beta$ and we obtain the second inequality in Theorem~\ref{thm:gibbsfe}. 
\end{proof}

\section{Canonical ensemble} \label{sec:canonical}

Here we prove Theorems~\ref{thm:fe} and~\ref{thm:lowdens}. They are deduced from their analogues in the constant-pressure ensemble, Theorems~\ref{thm:ge} and Theorem~\ref{thm:gibbsfe}.

\subsection{Proof of Theorem~\ref{thm:fe} } \label{sec:canoproof1}

We suppress the $\ell$-dependence from the notation, abbreviate $p_\beta = p(\beta,\ell)$, and note the relations
\be \label{eq:ellp}
	\ell = \frac{\partial g}{\partial p}(\beta,p_\beta), \qquad f(\beta,\ell) = g(\beta,p_\beta)- \ell p_\beta,
\ee
which follow from Eqs.~\eqref{eq:pfrel}, \eqref{eq:gfrel}, \eqref{eq:pflgrel} and standard results on Legendre transforms. 

Before we prove Theorem~\ref{thm:fe}, we formulate a simple lemma on convex functions and their Legendre transforms whose proof is omitted. 

\begin{lemma} \label{lem:lego}
	Suppose $\varphi, \varphi_1, \varphi_2, \ldots : \R \to \R\cup\{\infty\}$ are convex functions whose restrictions to an interval $(a, b) \subset \R$ are strictly convex and continuously differentiable. If $\varphi_n\to \varphi$ pointwise on $(a,b)$, then also $\varphi'_n\to \varphi'$ pointwise on $(a,b)$ and for all $y\in \varphi'((a,b))$,
	\[
		\lim_{n\to \infty} \varphi_n^*(y) = \varphi^*(y). 
	\] 
\end{lemma} 

Proof of (a): Recall $W(r) = \sum_{k=1}^m v(kr)$ and let $p^*:= |v(z_{\max})|/z_{\max}$. We apply Lemma~\ref{lem:lego} to $\varphi_\beta(p):= - g(\beta,p)$ and $\varphi(p):= - \inf_{r>0} (W(r) + p r) = W^*(-p)$ on the interval $(0,p^*)$, where $g(\beta,\cdot)$ and $W$ have been set to $+\infty$ for non-positive arguments. The function $\varphi_\beta$ is strictly convex and continuously differentiable because $p\mapsto g(\beta,p)$ is strictly concave and continuously differentiable, as noted in Section~\ref{sec-conspress}.  It follows from Assumption~\ref{ass:v1}(i)--(iv) that $W$ is strictly convex and smooth in $(z_{\min}, z_{\max})$, $W(z_{\max}) + p z_{\max} \leq v(z_{\max}) + p z_{\max} < 0$ for $p < p^*$ and $W(r)+pr > 0$ for $r \leq z_{\min}$ and $p \ge 0$. Thus, for $p\in [0, p^*]$, there is a unique $a(p)$ with $\varphi(p) = W(a(p)) + p a(p)$, and $a(p) \in (z_{\min},z_{\max})$ satisfies $W'(a(p)) + p =0$. Set $\ell^*:= a(p^*)$. Then $a(0) = a$ and $a(p)\in (\ell^*,a)$ for all $p\in (0,p^*)$. In particular, $\varphi$ is smooth and strictly convex on $(0, p^*)$ with $\varphi'(p) = - a(p)$. By \cite[Theorem 2.5]{jkst19}, we have
\[
	\lim_{\beta \to \infty} g(\beta,p) = \inf_{r>0} \bigl( W(r) + p r\bigr)
\] 
for all $p\in (0,p^*)$, hence $\varphi_\beta\to \varphi$ on $(0,p^*)$. Also notice 
\[
	\varphi_\beta ^*(-\ell) = \sup_{p>0} \bigl(- p \ell + g(\beta,p))\bigr) = f(\beta,\ell)
\] 
by~\eqref{eq:gfrel}. Lemma~\ref{lem:lego} thus implies $f(\beta,\ell) = \varphi^*(-\ell) = W^{**}(\ell) = W(\ell)$ for all $\ell \in (\ell^*, a)$. Another application of Lemma~\ref{lem:lego} in combination with \eqref{eq:pflgrel} then also yields 
\[ 
    p(\beta,\ell) 
    = - \frac{\partial f}{\partial \ell}(\beta,\ell) 
    \to - W'(\ell)
\]
for $\ell \in (\ell^*, a)$. 

Proof of (b): Pick $\ell>a$. Eq.~\eqref{eq:ellp}, Theorem~\ref{thm:gibbsfe} and the definition~\eqref{eq:qdef} of $q_\beta$ yield 
$$
	\beta p_\beta = \frac{\exp( -\beta e_\mathrm{surf}^\ssup{R}(\beta)/2)}{\sqrt{\ell - a}} (1+o(1))\qquad\mbox{and}\qquad q_\beta =  (\ell - a) \beta p_\beta (1+o(1))
$$
and 
$$
 	f(\beta,\ell)  = g(\beta,p_\beta) - p_\beta \ell 
 			= e_0^\ssup{R}(\beta) - p_\beta (\ell -a) - \frac{q_\beta}{\beta}(1+o(1)) +o(p_\beta) = e_0^\ssup{R}(\beta) -  2p_\beta (\ell -a)(1+o(1)). 
 			$$ 
We plug in the asymptotics of $p_\beta$ and obtain Theorem~\ref{thm:fe}(b). 

\subsection{Proof of Theorem~\ref{thm:lowdens}} \label{sec:canoproof2} 

Let $p>0$ and $n\in \N$. Recall $z_{N-1} =L-\sum_{j=1}^{N-2} z_j$.
Write $\boldsymbol{z} = (z_1,\ldots,z_{N-1})$.  Then 
$$
\begin{aligned} 
	\Q_N^\ssup{\beta}(M_N = n) 
		&  =\frac{1}{Q_N(\beta,p)}\int_{0}^\infty \e^{-\beta p L}\Bigl( \int_{\Delta_{N,L}} \1_{\{M_N=n\}}(\vect{z}) \e^{-\beta U_N(\vect{z})} \dd z_1\cdots \dd z_{N-2}\Bigr) \dd L \\
		&\geq \frac{1}{Q_N(\beta,p)}\int_{\ell N}^{\ell N +1} \e^{-\beta p L}\Bigl( \int_{\Delta_{N,\ell N}} \1_{\{M_N=n\}}(\vect{z}) \e^{-\beta U_N(\vect{z})} \dd z_1\cdots \dd z_{N-2}\Bigr) \dd L  \\
		&  \geq \frac{Z_N(\beta,\ell N)}{Q_N(\beta,p)} \e^{-\beta p (\ell N+1)}\, \P_{N,\ell N}^\ssup{\beta}(M_N = n). 
\end{aligned}
$$ 
Choosing $p= p_\beta= p(\beta,\ell)$, we have 
\[
	\lim_{N\to \infty}\frac{1}{\beta N}\log \Bigl( \frac{Z_N(\beta,\ell N)}{Q_N(\beta,p)} \e^{-\beta p (\ell N+1)} \Bigr) = - f(\beta,\ell) + g(\beta,p) - p \ell =0.
\]
Let $q_\beta$ be as in Eq.~\eqref{eq:qdef} and $\delta_\beta$, $c,C,\beta_0$ as in Theorem~\ref{thm:ge}. Then for $\beta\geq \beta_0$, 
\bes
\begin{aligned}
	&\limsup_{N\to \infty}\frac{1}{N} \log \P_{\beta,\ell N}^\ssup{\beta}\Bigl(\Bigl|\frac{M_N}{N} - q_\beta\Bigr|\geq C q_\beta \max\{q_\beta,\delta_\beta\} \Bigr) \\
        &\quad\leq \limsup_{N\to \infty}\frac{1}{N} \log \Q_{N}^\ssup{\beta}\Bigl(\Bigl|\frac{M_N}{N} - q_\beta\Bigr|\geq C q_\beta \max\{q_\beta,\delta_\beta\} \Bigr)	\leq - c q_\beta \delta_\beta^2. 
\end{aligned}
\ees
As $g_\mathrm{surf}^\cl(\beta,p_\beta) = e_\mathrm{surf}^\ssup{R}(\beta) + o(\beta^{-1})$ by Lemma~\ref{lem:qsurfcomp}, we have $q_\beta = (1+o(1)) q_{\beta,\ell}$. The first estimate in Theorem~\ref{thm:lowdens} follows. 

The statements on the empirical distributions are deduced in a similar fashion from the corresponding empirical distributions in Theorem~\ref{thm:ge}. For the geometric distributions, the proof is easily completed with the observation 
\bes
	\sum_{k=1}^\infty \Bigl| \frac{q_\beta}{(1+q_{\beta})^k} - \frac{q_{\beta,\ell}}{(1+q_{\beta,\ell})^k}\Bigr|\to 0. 
\ees

%%%%%%%% End stuff 

% % % %  bibliography
\bibliographystyle{amsalpha}
\bibliography{lenjon}

\end{document}